\documentclass{article}
\usepackage[accepted]{icml2021}

\usepackage{microtype}
\usepackage{graphicx}
\usepackage{subfigure}
\usepackage{booktabs} 

\usepackage{hyperref}

\usepackage[utf8]{inputenc} 
\usepackage[T1]{fontenc}    
\usepackage{hyperref}       
\usepackage{url}            
\usepackage{booktabs}       
\usepackage{amsfonts}       
\usepackage{nicefrac}       
\usepackage{microtype}      
\usepackage{graphicx}
\usepackage{amsmath,amsfonts,amssymb,amsthm}
\usepackage{array}
\usepackage{bm}
\usepackage{multicol}
\usepackage{lipsum}
\usepackage{wrapfig}
\usepackage{arydshln}
\usepackage{multirow}
\usepackage{balance}
\usepackage{dsfont}

\DeclareMathOperator*{\argmin}{arg\,min}

\newtheorem{lemma}{Lemma}
\newtheorem{remark}{Remark}
\newtheorem{theorem}{Theorem}
\newtheorem{corollary}{Corollary}
\newtheorem{definition}{Definition}

\newcommand{\Ebb}{\mathbb{E}}
\newcommand{\Rbb}{\mathbb{R}}
\newcommand{\Ncal}{\mathcal{N}}
\newcommand{\Hcal}{\mathcal{H}}
\newcommand{\Dcal}{\mathcal{D}}
\newcommand{\Xcal}{\mathcal{X}}

\newcommand{\Fcal}{\mathcal{F}}

\newcommand{\Vcal}{\mathcal{V}}
\newcommand{\Ccal}{\mathcal{C}}

\newcommand{\vbf}{\mathbf{v}}

\newcommand{\xbf}{\mathbf{x}}
\newcommand{\Xbf}{\mathbf{X}}
\newcommand{\hbf}{\mathbf{h}}
\newcommand{\Hbf}{\mathbf{H}}
\newcommand{\zbf}{\mathbf{z}}
\newcommand{\Zbf}{\mathbf{Z}}
\newcommand{\ybf}{\mathbf{y}}
\newcommand{\Ybf}{\mathbf{Y}}

\newcommand{\Wbf}{\mathbf{W}}

\newcommand{\ebf}{\mathbf{e}}
\newcommand{\dbf}{\mathbf{d}}
\newcommand{\Mbf}{\mathbf{M}}

\newcommand{\Ical}{\mathcal{I}}
\newcommand{\Lcal}{\mathcal{L}}
\newcommand{\Ucal}{\mathcal{U}}
\newcommand{\Ocal}{\mathcal{O}}
\newcommand{\Scal}{\mathcal{S}}

\newcommand{\tgamma}{\tilde{\gamma}}

\newcommand{\Pbb}{\mathbb{P}}
\newcommand{\wbf}{\mathbf{w}}
\newcommand{\Ubf}{\mathbf{U}}
\newcommand{\Vbf}{\mathbf{V}}

\newcommand{\Sigmabf}{\boldsymbol \Sigma}

\newcommand{\gammabf}{\boldsymbol \gamma}
\newcommand{\alphabf}{\boldsymbol \alpha}
\newcommand{\thetabf}{\boldsymbol \theta}
\newcommand{\Thetabf}{\boldsymbol \Theta}
\newcommand{\epsilonbf}{\boldsymbol \epsilon}
\newcommand{\etabf}{\boldsymbol \eta}

\newcommand{\thetabft}{\thetabf^{(t)}}

\begin{document}

\twocolumn[
\icmltitle{Rethinking Neural vs. Matrix-Factorization Collaborative Filtering: the Theoretical Perspectives}



\icmlsetsymbol{equal}{*}

\begin{icmlauthorlist}
\icmlauthor{Da Xu}{wmt}
\icmlauthor{Chuanwei Ruan}{insta}
\icmlauthor{Evren Korpeoglu}{wmt}
\icmlauthor{Sushant Kumar}{wmt}
\icmlauthor{Kannan Achan}{wmt}

\end{icmlauthorlist}

\icmlaffiliation{wmt}{Walmart Labs, Sunnyvale, California, USA.}
\icmlaffiliation{insta}{Instacart, San Francisco, California, USA. The work was done when the author was with Walmart Labs.}

\icmlcorrespondingauthor{Da Xu}{daxu5180@gmail.com}

\icmlkeywords{Neural collaborative filtering, Matrix factorization, Transducive machine learning, Inductive machine learning, Inverse propensity weighting, Data missing-not-at-random}
\vskip 0.3in
]



\printAffiliationsAndNotice{}  

\begin{abstract}
The recent work by \citet{rendle2020neural}, based on empirical observations, argues that matrix-factorization collaborative filtering (MCF) compares favorably to neural collaborative filtering (NCF), and conjectures the dot product's superiority over the feed-forward neural network as similarity function. In this paper, we address the comparison rigorously by answering the following questions: 1. what is the limiting expressivity of each model; 2. under the practical gradient descent, to which solution does each optimization path converge; 3. how would the models generalize under the inductive and transductive learning setting. Our results highlight the similar expressivity for the overparameterized NCF and MCF as kernelized predictors, and reveal the relation between their optimization paths. We further show their different generalization behaviors, where MCF and NCF experience specific tradeoff and comparison in the transductive and inductive collaborative filtering setting. Lastly, by showing a novel generalization result, we reveal the critical role of correcting exposure bias for model evaluation in the inductive setting. Our results explain some of the previously observed conflicts, and we provide synthetic and real-data experiments to shed further insights to this topic.
\end{abstract}

\section{Introduction}
\label{sec:introduction}
Neural collaborative filtering (NCF) and matrix factorization collaborative filtering (MCF) are the two major approaches for learning the user-item embeddings and combining them into similarity scores. From the early stage of collaborative filtering (CF), the factorization methods are actively explored thanks to their interpretability and computation advantage. \cite{koren2009matrix,koren2008factorization,rendle2010factorization}. In recent years, the success of deep learning in various other domains has motivated the adaptation of universal model architectures, such as the feedforward neural network (FFN), to CF tasks. A prevalent practice is to use FFN to combine the user and item embeddings, which is referred to as the NCF \cite{he2017neural}.   
As a consequence, the comparison between MCF and NCF was led by whether the similarity measurement should be computed by the dot product, or learnt by the FFN in a data-adaptive fashion. \citet{he2017neural} claims that NCF is superior because it is capable of expressing MCF under a particular parameterization. Since then, a significant proportion of new models use FFN as the default similarity function, and they are often reported with more competitive performances than MCF \cite{zhang2019deep}. Nevertheless, a recent paper argues based on empirical evidence that after training, FFN is unable to reconstruct the dot-product computation, and NCF can be outperformed by the MCF \cite{rendle2020neural}. 

It is not of surprise that the controversy arises since our understandings of both methods, specially NCF, are still somewhat limited from the theoretical standpoint. In particular, \citet{he2017neural} does not consider how training affects the parameters' dynamics, as well as how it alters the model expressivity in comparison with the dot-product formulation. Likewise, the conjectures in \citet{rendle2020neural} regarding the trained FFN are based on empirical observations so the conclusions may not generalize to other settings. On the other hand, while MCF has been studied intensively under the matrix completion framework with many theoretical guarantees \cite{candes2009exact, srebro2004maximum,srebro2005rank,candes2009exact}, it is nevertheless unclear how the inductive bias of using gradient descent training can affect the performance.

In addition to the choice of model, which is the main focus of the previous literature, we find that the testing performance is closely connected to how the tasks are designed, including the learning setting (inductive or transductive) as well as the relation between the empirical and exposure distribution. Towards that end, we summarize several major questions that prompt the ongoing debate: 

\textbf{1.} what is the expressive power of (what functions can be efficiently expressed by) each model, both at initialization and during training;

\textbf{2.} among the potentially many global optimum (due to the overparameterization), which solutions will be located by the practical gradient descent optimization; 

\textbf{3.} how would the training performances generalize to the unseen data which can be either obtained as random holdouts of a fixed finite population (the transductive setting) or i.i.d samples according to some underlying distribution (the inductive setting).


We point out that the third question is particularly complicated for CF. When viewing CF as a matrix completion problem, the training and testing data are constructed via random splits of the finite matrix index set, and are thus fixed before training. On the other hand, when treating the testing data as new i.i.d samples, the exposure bias \cite{chen2020bias} causes the discrepancy between: how the training data are sampled in practice (often according to the empirical data distribution), and which testing distribution to generalize (users are more likely to interact with items that are exposed to them). As a consequence, the same model can show various generalization behaviors in different settings.  

Answering each question requires in-depth understandings for the underlying modelling, optimization and generalization mechanisms. This paper also aims to explain how the different factors contribute to the controversies in the previous literature. The recent developments in the neural tangent kernel (NTK) provide powerful tools to study the limiting behavior of overparameterized models \cite{jacot2018neural}. We leverage them to first study the expressivity of NCF and MCF and reveal their relations via the lens of kernelized predictors. By studying the local updates in the space of parameters, we then provide a unified analysis on the inductive bias of gradient descent \cite{soudry2018implicit} for both NCF and MCF. We find that both optimization paths converge in direction to some norm-constraint max-margin solutions. While our result is a restatement for the FFN \cite{gunasekar2018implicit}, it is novel for MCF and connects the two cases.
By further deriving the (provably tight) norm-constraint capacities of NCF and MCF, we show their different generalization behaviors under the transductive and inductive learning setting. We also reveal a novel generalization result that accounts for the discrepancy between the empirical data and exposure distribution. We conclude our contributions as follow.

\textbf{1.} We reveal the limiting expressivity of NCF and MCF as kernelized predictors under a specific type of kernel that conforms to the collaborative filtering principle.

\textbf{2.} We provide a unified analysis that shows for NCF and MCF the gradient descent converge to the corresponding max-margin solutions.

\textbf{3.} We study and compare the transductive and inductive generalization behaviors for NCF and MCF, and provide a novel generalization result for the inductive setting that accounts for the impact of exposure.

\textbf{4.} We provide experiments on synthetic and real-data to examine and illustrate our results in practical applications.

Compared with \citet{rendle2020neural}, we reveal the complex dynamics of model expressivity, optimization and generalization that underlies the comparison between MCF and NCF. Our results rigorously justify the conclusions drawn under each perspective and problem instance.





\section{Related Work}
\label{sec:related-word}
The kernel induced by deep learning models, specially the neural tangent kernel (NTK), has been identified as the key component to study the limiting behavior of overparameterized neural networks \cite{daniely2016toward,jacot2018neural,yang2019scaling,chizat2019note} and connect the gradient flow training to the reproducing kernel Hilbert space (RHKS) \cite{arora2019fine,mei2019mean}. 

Nevertheless, the limiting-width (scale) results may not fully explain the success of training practical models to zero loss and still being able to generalize \cite{zhang2016understanding}. Towards that end, considerable efforts are spent studying the optimization geometry and the inductive bias of gradient descent \cite{soudry2018implicit}, and the convergence results have been shown for various FFN models \cite{lyu2019gradient,gunasekar2018implicit,chizat2020implicit}. However, the analogous results for matrix factorization are mostly established for the the matrix sensing problem under the squared loss \cite{gunasekar2018implicit,arora2019implicit,li2018algorithmic}, and are therefore not application to the implicit recommendation setting\footnote{The users express their interests implicitly, e.g. via click or not click, instead of providing explicit ratings.}, which is the focus of this paper. 

Finally, while the generalization results of FFN and matrix factorization have been studied previously \cite{wei2019regularization,srebro2004maximum, srebro2005rank}, their settings and assumptions are different from the collaborative filtering tasks (see Section \ref{sec:generalization} for detail), and are therefore less informative for analyzing NCF and MCF. On the other hand, despite the increasing recognition of adjusting for the exposure bias when training and evaluating recommenders \cite{xu2020adversarial,schnabel2016recommendations,liang2016modeling,yang2018unbiased,gilotte2018offline}, it is unclear how the common re-weighting approach affects generalization. This line of research also relates to the data missing-not-at-random (MNAR) domain, where the inverse propensity weighting method is heavily investigated in the context of recommender systems \cite{yang2018unbiased,saito2020unbiased}. 

\section{Notations and Background}
\label{sec:background}
To be consistent with the reference work of \citet{rendle2020neural}, we also adopt the implicit feedback setting which is more practical for the modern recommender systems. 

We use $u \in \Ucal$ and $i \in \Ical$ to denote the user and item, and use $\Dcal = \{(u,i) \, | \, u \in \Ucal, i \in \Ical\}$ to denote the collected user-item interactions. The vectors, random variables, and matrices and denoted by the lower-case bold-font, upper-case, and upper-case bold-fond letters. The vector $\ell_2$ norm, matrix Frobenius norm and nuclear norm are given by $\|\cdot\|_2$, $\|\cdot\|_F$, and $\|\cdot\|_*$. The sigmoid function is given by $\sigma(\cdot)$.

We use $y_{u,i} \in \{-1, +1\}$ to denote the implicit user feedback, $o_{u,i} \in \{0,1\}$ to indicate if the item has been exposed to the user, and $P_{O}$ to denote the exposure distribution: $p(O_{u,i}=1)$. The empirical data distribution is denoted by $P_n$. The classification model (predictor) is selected from $\Fcal = \{f(\thetabf;\, \cdot) \, | \, \thetabf \in \Thetabf\}$. We use $F(\thetabf)$ to denote the \emph{decision boundary} associated with $f(\thetabf;\, \cdot)$, and use $\ell(\cdot)$ to denote the loss function. Both MCF and NCF use embeddings to represent user and items\footnote{The user and item embeddings are allowed to have different dimensions in NCF. Here, we assume they have the same dimensions to be comparable to MCF.}, which we denote by $\zbf_u, \zbf_i \in \Rbb^d$.
Throughout this paper, we assume the models are overparameterized, so analytically there exists parameterizations $\thetabf^*$ that perfectly separates the data: $\forall (u,i) \in \Dcal: y_{u,i} f(\thetabf^*;(u,i)) \geq 0$.

\subsection{Dot product and learnt similarity}
\label{sec:NCF-MCF}

The general formulation of MCF is given by\footnote{Some previous work reformulate MCF as: $\beta_u + \beta_i + \langle \zbf_u , \zbf_i \rangle$, which is a special case of the general formulation.}:
\[
f^{\text{MCF}}(u,i) := f\big((\Zbf_U, \Zbf_I);(u,i)\big) = \langle \zbf_u , \zbf_i \rangle,
\]
where $\Zbf_U$ and $\Zbf_I$ are matrix collections of the user and item embeddings. When viewed as the matrix completion problem, MCF relies only on the dot product to compute the similarity and recover the observed user-feedback matrix $\Ybf \in \Rbb^{|\Ucal| \times |\Ical|}$ by minimizing: $\Lcal \big (\Zbf_u, \Zbf_i \big ) = \frac{1}{|\Dcal|}\sum_{(u,i)\in \Dcal} \ell \big( \big\langle \Ybf_{u,i} , \big[\Zbf_U \Zbf_I^{\intercal}\big]_{u,i} \big\rangle \big)$.

The NCF, on the other hand, employs a feed-forward neural network $\phi(\thetabf^{\text{NN}};\, \cdot)$ to learn the similarity of the user and item by first combining their embeddings via addition or concatenation\footnote{The original paper of \citet{he2017neural} suggests using the concatenation, but in our experiments we find adding the embeddings sometimes provide superior performances so we study them both.}: $f^{\text{NCF-a}}(u,i) := \phi\big(\thetabf^{\text{NN}};\zbf_u+\zbf_i\big)$ and $f^{\text{NCF-c}}(u,i) := \phi(\thetabf^{\text{NN}};[\zbf_u, \zbf_i]\big)$. Without loss of generality, we assume the following multi-layer percepton (using the addition $\xbf:=\zbf_u+\zbf_i$ as an example): 
\[\phi(\thetabf^{\text{NN}};\xbf) = \Wbf_1 \tilde{\sigma}(\Wbf_2 \xbf), \, \Wbf_1 \in \Rbb^{1\times d_1}, \, \Wbf_2 \in \Rbb^{d_1\times d},\] 
where $\tilde{\sigma}(\cdot)$ is the ReLU activation. The risk is then given by: $\Lcal\big(\theta^{\text{NN}},\Zbf_u, \Zbf_i\big) = \frac{1}{|\Dcal|}\sum_{(u,i)\in \Dcal} \ell\big(y_{u,i}, f^{\text{NCF}}(u,i) \big)$.

\subsection{Transductive and inductive learning for collaborative filtering}
\label{sec:trans-induc}

Transductive and inductive learning differ primarily in whether the learnt patterns shall generalize to a specific testing data dependent on the training data, or the samples with respect to some distribution of the population from which the training data is also drawn \cite{vapnik1982estimation,olivier2006semi}. For CF in practice, the two learning settings may appear in different scenarios, and each setting can be found suitable by particular tasks. 

Specifically for collaborative filtering, transductive learning treats the training and testing $(u,i)$ pairs as obtained in advance via a random split of the finite $\Dcal$, which is equivalent to \emph{sampling without replacement} and thus the dependency between the training and testing data. Compared with the inductive learning whose testing samples are drawn from some unknown distribution over the $(u,i)$ indices, transductive learning aims at a specific subset of $(u,i)$ pairs. The traditional CF tasks, such as movie rating and item-to-item recommendation, are better characterized by transductive learning because their goal is to predict well on a specific subset of entries whose indices are also known in advance.

However, for the negative-sampling-based training and evaluation approach, which often applies to the modern large-scale problems, the data is characterized by an incoming stream of i.i.d samples. As such, inductive learning appears to be the more appropriate setting, yet there still exists a critical gap caused by the exposure bias: the underlying distribution is often shifted by how the items were exposed to users. Inversely weighting each sample by its exposure probability gives the unbiased estimator as if the exposure was \emph{uniformly as random} \cite{schnabel2016recommendations}, i.e.
\begin{equation*}
\begin{split}
& \Ebb_{P_{O}}\Big[\frac{1}{|\Dcal|}\sum_{(u,i)\in\Dcal}\frac{\ell(y_{u,i}, \hat{y}_{u,i})}{p(O_{u,i}=1)}\Big] \\
& = \frac{1}{|\Dcal|}\sum_{(u,i)\in\Dcal} \frac{\ell(y_{u,i}, \hat{y}_{u,i})}{p(O_{u,i}=1)}p(O_{u,i}=1) \\
& = \Ebb_{\hat{P}_{\text{unif}}}[\ell(Y, \hat{Y})]. 
\end{split}
\end{equation*}
The uniform exposure scenario is ideal because the training and testing will not be affected by the previous exposure, thus leading to fair evaluation and comparison. 
Clearly, the inverse weighting approach can effectively correct for the distribution shift caused by the exposure bias if $p(O_{u,i}=1)$ is known in advance. We will discuss in Section \ref{sec:summary} the limitation of this presumption. Another view of the reweighting method is to think of the testing distribution as $P_\text{unif}$, which is uniform on $\Dcal$, and the training distribution as $P(O_{u,i}=1)$, which corresponds to the previous exposure.


\subsection{Experiment Setup}
\label{sec:experiment}
We use the \emph{MovieLens-100K}\footnote{https://grouplens.org/datasets/movielens/100k/} dataset for the real-data and semi-synthetic experiments. Following \citet{rendle2020neural} and \cite{he2017neural}, we first covert the data to implicit feedback by treating the positive ratings as clicks.
The primary reasons for using this relatively smaller dataset (which consists of the 100,000 ratings from 1,000 users on 1,700 movies) is the concerning computation feasibility:

1. to show the limiting expressivity and convergence results, we often need to compute the ground truth by solving convex optimization problems exactly and run the gradient descent optimization for several thousands of epochs;

2. a recent paper points out that ranking metrics computed via subsampling can be misleading in practice \cite{rendle2019evaluation}, so we conduct full scans for metric computation.

The MCF and NCF are implemented with \emph{Tensorflow} using the stochastic gradient descent optimizer and log loss: $\ell(u) = \log(1+\exp(-u))$. All the ranking metrics, i.e. ranking AUC, top-k hitting rate (HR@k) and normalized discounted cumulative gain (NDCG@K), are computed after ranking all the movies for the user. Following \cite{rendle2020neural,he2017neural}, for each user, we use the last interaction for testing and the rest for training and validation, and equip each positive sample with four negative samples unless specified otherwise. The detailed settings for each section are relegated to the appendix. 
All the reported results are computed from ten repetitions.

\section{The Limiting Expressivity and Collaborative Filtering Kernel}
\label{sec:kernel}
The key challenge for directly comparing the expressivity of dot product and learnt similarity is that the parameters are constantly changing during the gradient descent update: $\thetabf^{(t)} = \thetabf^{(t-1)} - \eta \nabla \Lcal(\thetabf^{(t-1)})$. The different training dynamics and the randomness in initialization prohibit the shoulder-to-shoulder comparison, unless we can find some key factors invariant to training and decisive for model expressivity. A feasible direction is toward the infinite-width domain\footnote{The width of each layer tends to infinity, i.e. $d, d_1 \to \infty$ for the models in Section \ref{sec:NCF-MCF}. Here, we assume $d_1=d$, w.l.o.g.}, where the overparameterized models are observed to converge to zero training loss while their parameters hardly vary \cite{allen2019convergence,du2019gradient}.
The invariant factor that is recognized to fully describe the training dynamics is the neural tangent kernel (NTK) \cite{jacot2018neural,arora2019exact}:
\begin{equation*}
    K\big( (u,i), (u',i')  \big) = \big\langle \nabla f\big(\thetabf; (u,i)\big) , \nabla f\big(\thetabf; (u',i')\big) \big\rangle,
\end{equation*}
which is a fixed quantity that depends only on the model architecture as $d\to\infty$. We refer to the model expressivity under the infinite-width NTK as its limiting expressivity.

To see the critical role of a fixed NTK for model expressivity, note that by Taylor expansion, infinite-width MF and FFN behave like a linear function around their \emph{scaled initializations} (i.i.d with $N(0, \alpha/d)$ for some constant $\alpha$):
\begin{equation*}
\begin{split}
f\big(\thetabf;(u,i)\big) & = f\big(\thetabf^{(0)};(u,i)\big) + \\
&  \Big \langle \thetabf - \thetabf^{(0)} , \nabla f\big(\thetabf^{(0)};(u,i)\big) \Big \rangle + \Ocal\big(\sqrt{1/d}\big).
\end{split}
\end{equation*}

Analytically, we can always use reparameterization to discard the intercept term $f\big(\thetabf^{(0)};(u,i)\big)$ by finding $f(\thetabf;\, \cdot) = g(\thetabf_1;\, \cdot) - g(\thetabf_1;\, \cdot)$ and letting $\thetabf_1^{(0)} = \thetabf_2^{(0)}$. Therefore, the limiting expressivity of $f\big(\thetabf; (u,i)\big)$ is well-approximated by $\Big \langle \thetabf - \thetabf^{(0)} , \nabla f\big(\thetabf^{(0)};(u,i)\big) \Big \rangle$ as $d \to \infty$. It is now a straightforward analogy that $\nabla f\big(\thetabf^{(0)};(u,i)\big)$ plays the role of the feature lift that maps all the $(u,i)$ pair to their representations in the reproducing kernel Hilbert space (RHKS) induced by the NTK, i.e. $K: (\Ucal, \Ical) \times (\Ucal, \Ical) \to \Rbb$ \cite{shawe2004kernel}. Consequently, under the infinite-width limit, applying gradient descent in the original function space is equivalent to finding the optimal RHKS predictor.

We point out that the RHKS view of MCF and NCF during optimization is a natural continuation of their expressive power at initialization. At initialization, the dot-product formulation of MCF directly characterizes the Gram matrix for the user-item similarity kernel \cite{paterek2007improving}, and standard FFN is known to concentrate to a Gaussian process characterized by its covariance kernel \cite{neal1996bayesian,yang2019scaling}. In the sequel, we wish to understand the relation between the NTK of MCF and NCF, and find out which solutions are optimal in the corresponding RHKS. Interestingly, both MCF and NCF lead to a specific type of kernel that we believe interprets collaborative filtering in principle. We summarize the results as below.

\begin{theorem}
\label{thm:expressivity}
Let $\Dcal_{\text{train}} \subseteq \Dcal$ be the training data. Under the exponential loss $\ell(u) = \exp(-u)$ or log loss $\ell(u) = \log(1+\exp(-u))$, by applying gradient descent with small learning rate and scaled initializations, the decision boundary for infinite-width MCF and NCF connects to the minimum-norm RHKS predictor via:
\begin{equation*}
\label{eqn:thm1}
\begin{split}
 & \lim_{t\to\infty} \lim_{d\to\infty} F\Big(\frac{\thetabf^{(t)}}{\|\thetabf^{(t)}\|_2}\Big) \overset{\text{stationary points of}}{\to} \\ &   \Big\{\argmin_{f:(\Ucal, \Ical)\to \Rbb} \big\|f\big\|_{K_{\text{CF}}} \, s.t. \, y_{u,i}f(u,i) \geq 1, \, \forall (u,i)\in \Dcal_{\text{train}} \Big\},
\end{split}
\end{equation*}
where $\|\cdot\|_{K}$ is the induced RKHS norm and $K_{\text{CF}}$ is the \textbf{collaborative filtering kernel} parameterized via:
\[
K_{\text{CF}}\big((u,i) , (u',i') \big) = a + b\cdot 1 [i=i'] + c \cdot 1[u=u'],
\]
for some $a, b, c \geq 0$ when $(u,i) \neq (u',i')$. 
\end{theorem}

\begin{figure}[htb]
    \centering
    \includegraphics[width=0.95\linewidth]{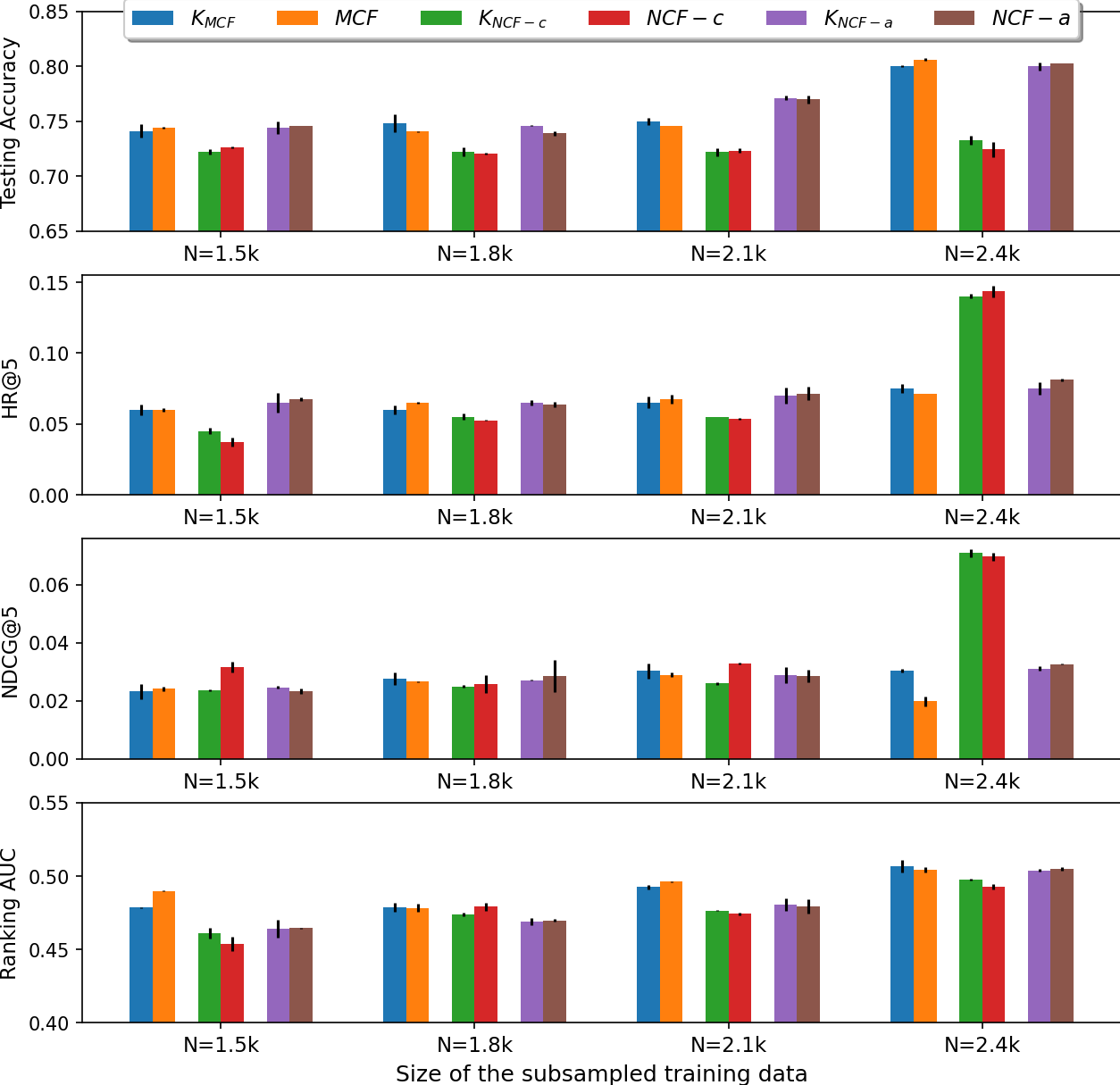}
    \caption{The comparison between the RKHS predictors in Theorem \ref{thm:expressivity} and their corresponding MCF or NCF methods on the testing data. We use \emph{LIBSVM} \cite{chang2011libsvm} for the kernelized support vector machine (SVM). The subsampling is performed by sampling 200 items according to the popularity, and sampling 200 users uniformly at random. For each $(u,i)$ pair with a positive feedback, we sample different number of negative instances without replacement to obtain the desired sample size of $N$. We set $d=128$ and $d_1=64$ for MCF and NCF to approximate the large-width setting as the dataset is very small. To be consistent with the theory, we use the standard scaled normal initialization and train 3,000 epochs under the constant learning rate $0.01$.}
    \label{fig:kernel-cf}
\end{figure}

For instance, under the standard scaled initialization $N(0, 1/d)$, we have $a=0$, $b=c=1$ for $K_{\text{MCF}}$; $a=1/\pi$, $b=c=\frac{1}{2} - \frac{1}{2\pi}$ for $K_{\text{NCF-c}}$, and $a=1/\pi$, $b=c=\frac{1}{2} - \frac{(2-\sqrt{3})}{2\pi}$ for $K_{\text{NCF-a}}$. The details and proofs are relegated to the appendix. 

The significance of Theorem \ref{thm:expressivity} is that we now have a dual charactierzation of model expressivity for the \emph{trained} MCF and NCF via the lens of RKHS predictors. Furthermore, $K_{\text{CF}}$ provide a natural kernel interpretation for the CF principle by emphasizing how predictions are made from the collaborative user or item signal. Recall from the standard SVM arguments that the limiting classifier in Theorem \ref{thm:expressivity} is given by:
\[
f(u',i') = \sum_{(u,i)\in \Dcal_{\text{train}}}\alpha_{u.i} y_{u,i} K_{\text{CF}}\big((u,i) , (u',i') \big),
\]
where $\alpha_{u,i}\geq 0$ are the corresponding dual variables \cite{drucker1997support}.
Specifically for $K_{\text{CF}}$, the value of $a$, which gives the kernel value when $u\neq u'$ and $i\neq i'$, captures the global (population) bias; $b$ and $c$ capture the relative importance\footnote{We illustrate in the appendix on how the relative magnitude of $a,b,c$ can be changed by the initializations.} of neighboring data points who share the same user or item. In Figure \ref{fig:kernel-cf}, we compare the testing results of large-width MCF, NCF and their $K_{\text{CF}}$-induced SVM on the sampled subsets of the MovieLens data. The subsampling is necessary because the kernel method has a space complexity of $\Ocal\big(|\Dcal| \times |\Dcal|\big)$. We observe highly comparable results from the trained CF and their corresponding RHKS predictor on testing data, which suggests their similar limiting performances \footnote{The results does not imply that MCF and NCF can be effectively replaced in practice because the kernel methods are impractical even for moderate recommendation dataset.}.

Therefore, regardless of the distinguished formulations of dot product and learnt similarity, both MCF and NCF are capable of expressing the same type of RKHS predictor in their limit. The only difference lies in their specific kernel parameterization, which immediately implies that each model's limiting performance can depend on the specific data distribution. In particular, \citet{bartlett2002rademacher} show that the inductive generalization performance of the RKHS predictor in Theorem \ref{thm:expressivity} satisfies:
\[ 
\text{Test error} \leq \text{Train error} + \Ocal\bigg(\sqrt{\frac{ \alphabf^{\intercal} K_{\text{CF}} \alphabf \cdot \text{trace}(K_{\text{CF}})}{|\Dcal_{\text{train}}|}} \bigg),
\]
where $\alphabf$ is the vector of dual variables which are data-dependent. Therefore, it is impossible to draw an exclusively comparison between the two methods, because each dataset can cause a particular $\alphabf$ that lead to different comparisons of $\alphabf^{\intercal} K_{\text{CF}} \alphabf$. 
Consequently, our results suggest that the comparison between MCF and NCF should not be concluded sightly just based on their expressivity, at least not in the limiting sense. We still need to understand how the parameterization interacts with gradient descent under the practical model size and initialization, which are the topics of the next section.

\section{Optimization Paths of MCF and NCF under Gradient Descent}
\label{sec:max-margin}
The key observation that motivates this section is that when trained with gradient descent on the MovieLens data without using any explicit regularization (Figure \ref{fig:transductive}), the classification and ranking metric on testing data keep increasing while the training loss decreases to zero. Our observation is consistent with \citet{zhang2016understanding} and the follow-up analysis of \citet{soudry2018implicit,gunasekar2019implicit}, that gradient descent carries the inductive bias of selecting the global optimum that tends to generalize well. Since this phenomenon may hold for any model size and initialization as pointed out by \citet{soudry2018implicit}, the underlying mechanism may provide a more complete understanding of the trained MCF and NCF in addition to the previous section. 

Also, in Figure \ref{fig:transductive}, we show another phenomenon that the norms of parameter matrices diverge as the loss converges to zero. It is not a surprising result since the gradient-descent training is conducted under the following conditions\footnote{The more refined discussions for each condition, as well as their implications, are provided in the appendix.}:

\textbf{C1}. The loss function has the exponential-tail behavior such as the exponential loss and log loss;

\textbf{C2}. Both the MCF and NCF in Section \ref{sec:NCF-MCF} are $L$-homogeneous, i.e. $f(\thetabf;\,\cdot) = \|\thetabf\|_2^L \cdot f\big(\thetabf/\|\thetabf\|_2;\,\cdot\big)$ for some $L>0$, and have some smoothness property;

\textbf{C3}. The data is separable with respect to the overparameterized MCF and NCF (introduced in Section \ref{sec:background}).

For clarity purpose, we consider the exponential loss $\ell(u) = \exp(-u)$. Setting aside the technical details, the heuristic explanation is that with a proper learning rate, gradient descent on smooth objectives converges to the stationary points (gradient vanishes). Since the gradient is given by:
\[ 
\frac{1}{|\Dcal|} \sum_{(u,i)\in \Dcal} -y_{u,i} \exp\big(-y_{u,i} f\big(\thetabf; (u,i)\big)\big) \nabla f\big(\thetabf; (u,i)\big),
\]
disregarding the corner cases where the gradient terms are linear dependent, the necessary condition for a zero gradient is: $\forall (u,i)\in\Dcal: y_{u,i}f\big(\thetabf; (u,i)\big) \to \infty$, which implies that $\|\thetabf\|_2 \to \infty$ since $f$ is homogeneous.
The norm divergence plays an important role here because the loss function is now dominated by: $\min_{(u,i)\in \Dcal} y_{u,i} f\big(\thetabf; (u,i)\big)$, which corresponds to the worst margin, due to the fact that the loss function has an exponential behavior. Therefore, as the parameter norm diverges, the direction of the optimization path resembles that of the hard-margin SVM (shown as below) since $\Dcal$ is separable.
\[
\max_{\thetabf: \|\thetabf\|_2\leq 1} \min_{(u,i)\in \Dcal} y_{u,i} f\big(\thetabf; (u,i)\big),
\] 
The same reasoning holds for MCF under the nuclear norm, since it relates to the $\ell_2$ norm via: $\|\Xbf\|_* = \frac{1}{2}\min_{\Ubf \Vbf^{\intercal}=\Xbf}(\|\Ubf\|_F^2 + \|\Vbf\|_F^2)$. 

Providing the rigorous arguments for the above heuristic is nontrivial since $f(\thetabf;\, \cdot)$ can be non-convex for NCF, so the connection between stationarity and optimality can be tricky. We refer to a similar idea in \citet{lyu2019gradient}, which bridges the discrepancy by showing a specific type of constraint qualification that leads to the KKT points. We summarize our results as below.

\begin{theorem}
\label{thm:optimization}
Under condition C1, C2 and C3, with a small constant learning rate, $\thetabf^{\text{NCF}}:=[\thetabf^{\text{NN}}, \Zbf_U, \Zbf_I]$ converges in direction, i.e. $\lim_{t\to\infty} \thetabf^{(t)} / \|\thetabf^{(t)}\|_2$, to the KKT point of:
\begin{equation*}
    \min \, \|\thetabf\|_2 \quad s.t. \quad y_{u,i}f^{\text{NCF}}\big(\thetabf; (u,i)\big) \geq 1, \, \forall (u,i)\in \Dcal_{\text{train}}.
\end{equation*}
The predictor of MCF: $\Xbf = \Zbf_U \Zbf_I^{\intercal}$ converges in direction, i.e. $\lim_{t\to\infty} \Xbf^{(t)} / \|\Xbf^{(t)}\|_*$, to the stationary point of:
\begin{equation}
\label{eqn:nuc-svm}
    \min \, \|\Xbf\|_* \quad s.t. \quad y_{u,i}\Xbf_{u,i} \geq 1, \, \forall (u,i)\in \Dcal_{\text{train}}.
\end{equation}
\end{theorem}

Theorem \ref{thm:optimization} extends the results in \citet{soudry2018implicit,lyu2019gradient,gunasekar2019implicit} to matrix factorization setting, and complements \citet{gunasekar2018implicit,arora2019implicit} who study matrix sensing under the squared loss. We also provide a unified proof in the appendix. 

\begin{figure}[hbt]
    \centering
    \includegraphics[width=\linewidth]{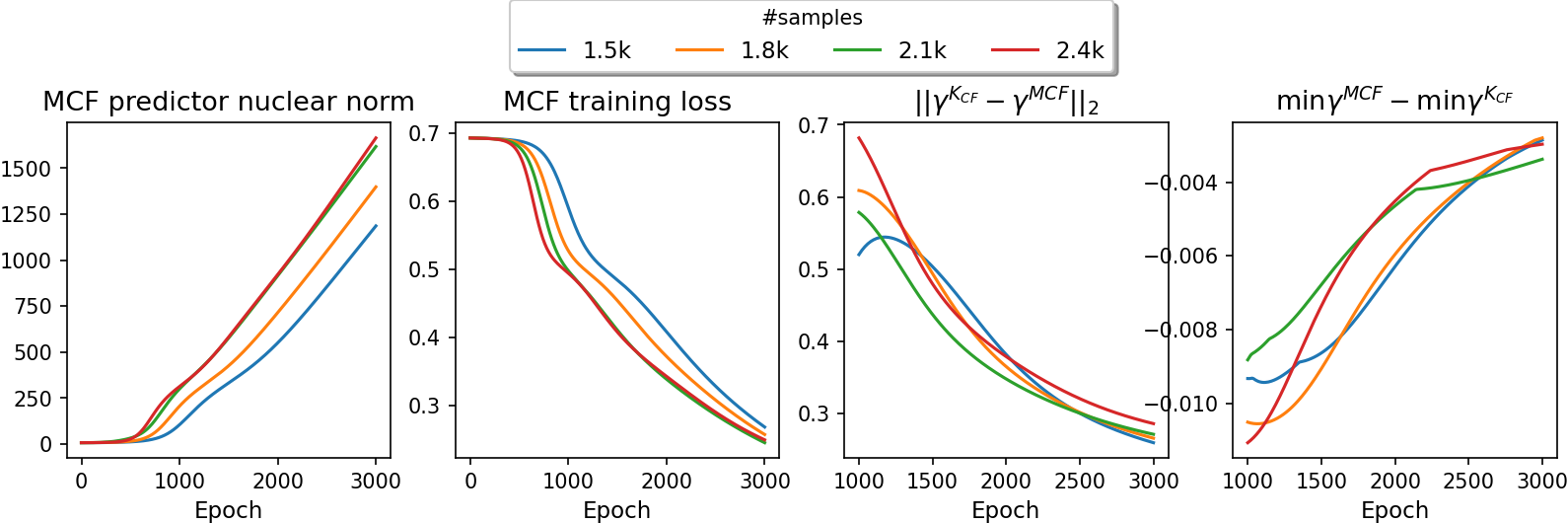}
    \caption{We adopt the sub-sampled dataset described in Figure \ref{fig:kernel-cf}. We first use the \emph{CVX solver} \cite{grant2014cvx} to obtain the exact solutions of the convex programming of e.q. (\ref{eqn:nuc-svm}). In contrast to the optimization setup of Figure \ref{fig:kernel-cf}, here we consider the more practical scenario by using the unscaled $N(0,0.1)$ initialization, the moderate width of $d=32$ and the learning rate of $0.1$. We do not show the error bar in the line charts for visualization clarity.}
    \label{fig:SVM_MCF}
\end{figure}

The optimization paths described in Theorem \ref{thm:optimization} reveal an intrinsic difference between MCF and NCF in terms of which minimum-norm solutions are tended by the gradient descent in a practical setting. Since the corresponding hard-margin SVM problem for MCF (state in (\ref{eqn:nuc-svm})) is convex, we are able to to obtain its exact global minimum, and use it as an oracle to compare with the actual behavior of MCF during training. We particularly focus on the margins of the normalized (by nuclear norm) predictor, which we denote by $\gammabf$. In Figure \ref{fig:SVM_MCF}, we first verify that the nuclear norm divergence indeed accompanies the convergence of training loss. Crucially, both $\|\gammabf^{K_{\text{CF}}} - \gammabf^{\text{MCF}}\|_2$ and $ \min \gammabf^{\text{MCF}} - \min \gammabf^{K_{\text{CF}}}$ converge (by decreasing and increasing) to zero, as anticipated by the conclusions of Theorem \ref{thm:optimization}.

Compared with Section \ref{sec:kernel}, the results here reveal optimization path as a more likely factor that distinguishes MCF and NCF in practice. The nuclear-norm and $\ell_2$-norm constraint optimization paths imply the different capacities of the predictor, whose impact on the generalization performance is among the core investigation in the learning theory. We provide the detailed characterization and their comparisons in the next section.

\section{Generalization for Transdutive and Inductive Collaborative Filtering}
\label{sec:generalization}
Before we state the main results, we point out that our setting covers a broad range of settings for MCF and NCF. Firstly, we consider the neural network in NCF to be any $q-$layer FFN, i.e. parameterized by $[\Wbf_1,\ldots,\Wbf_q, \Zbf_U, \Zbf_I]$, with ReLU activation. Secondly, we allow the explicit $\ell_2$-norm regularizations, which are sometimes found useful to improve model performance under finite-step training. Together with the results from Section \ref{sec:max-margin}, we expect that both MCF and NCF will tend to the norm-constraint solutions, either by training toward convergence without the explicit norm regularization or by training under the explicit norm regularizations. 

As we mentioned in Section \ref{sec:trans-induc}, the transductive and inductive CF requires studying different generalization mechanisms, so MCF and NCF may compare differently in each setting. For inductive CF, the exposure distribution often decides which $(u,i)$ pairs are more likely to appear in the training data, so its discrepancy with the idea uniform-exposure scenario under $P_{\text{unif}}$ is also a critical factor for a fair comparison of MCF and NCF. For transductive CF, the goal is to bound the classification error on the given $\Dcal_{\text{test}}$, i.e.  $\text{Err}_{\Dcal_{\text{test}}}(f):=\sum_{(u,i) \in \Dcal_{\text{test}}} 1\big[ y_{u,i} f(u,i) \leq 0 \big]$, while for the inductive CF we wish to bound $\Ebb_{(u,i) \sim P_{\text{unif}}} 1\big[ y_{u,i} f(u,i) \leq 0 \big]$ under the ideal uniform-exposure scenario that we explained in Section \ref{sec:background}. For many real-world problems, and in the MovieLens dataset that we are experimenting with, the number of items far exceeds the number of users. So we assume $|\Ucal| < |\Ical|$ without loss of generality. 

For the transductive CF, the testing data is often made proportional to the training data so here we assume $|\Dcal_{\text{test}}| = \beta |\Dcal_{\text{train}}|$ for some $0<\beta<1$. Notice that $\Dcal_{\text{test}} \cap \Dcal_{\text{train}} = \emptyset$ due to the sampling without replacement. Also, $\Dcal = \Dcal_{\text{test}}\cup \Dcal_{\text{train}}$. We first state the generalization result for transductive CF in terms of the classification error.

\begin{theorem}[\textbf{Transductive CF}]
\label{thm:transductive}
Let $f^{\text{MCF}},\,f^{\text{NCF}}: \Ucal\times\Ical\to\Rbb$ be the predictor for MCF and NCF.
Suppose that $\max_{(u,i)\in \Dcal}\|\zbf_u\|_2 + \|\zbf_i\|_2 \leq B_{\text{NCF}}$ for NCF with concatenation, $\max_{(u,i)\in \Dcal}\|\zbf_u+\zbf_i\|_2 \leq B_{\text{NCF}}$ for NCF with addition, and $\|\Wbf_i\|_F \leq \lambda_i$ for $i=1,\ldots,q$. Also suppose for MCF that $\big\|\Zbf_U\Zbf_I^{\intercal}\big\|_{*}\leq \lambda_{\text{nuc}}$. 
Then with probability at least $1-\delta$ over the random splits of $\Dcal$, for any $\big\{f^{\text{NCF}}(u,i) \, \big | \, (u,i)\in\Dcal\big\}$ and $\gamma>0$:
\begin{equation*}
\begin{split}
    & \text{Err}_{\Dcal_{\text{test}}}(f^{\text{NCF}}) \leq \sum_{(u,i) \in \Dcal_{\text{train}}} 1\big[ y_{u,i} f^{\text{NCF}}(u,i) \leq \gamma \big] + \\
    & \quad \Ocal\bigg( \frac{(1+\beta)\sqrt{q}B_{\text{NCF}} \prod_{i=1}^q \lambda_i}{\gamma \beta |\Dcal_{\text{train}}|^{1/2}}  \bigg) + \Ocal\bigg( \sqrt{\frac{\log 1/\delta}{|\Dcal_{\text{test}}|}} \bigg),
\end{split}
\end{equation*}
and for any $\big\{f^{\text{MCF}}(u,i) \, \big | \, (u,i)\in\Dcal\big\}$ and $\gamma>0$:
\begin{equation*}
\begin{split}
    & \text{Err}_{\Dcal_{\text{test}}}(f^{\text{MCF}}) \leq \sum_{(u,i) \in \Dcal_{\text{train}}} 1\big[ y_{u,i} f^{\text{MCF}}(u,i) \leq \gamma \big] + \\
    & \quad \Ocal\bigg( \frac{(1+\beta)\sqrt[4]{\log |\Ucal|} \sqrt{|\Ical|}\cdot \lambda_{\text{nuc}}}{\gamma \beta |\Dcal_{\text{train}}|}  \bigg) + \Ocal\bigg( \sqrt{\frac{\log 1/\delta}{|\Dcal_{\text{test}}|}} \bigg).
\end{split}
\end{equation*}
\end{theorem}

Theorem \ref{thm:transductive} gives the first NCF generalization result in the transductive setting, and while the generalization of matrix sensing has been studied in such as \citet{srebro2005rank} and \citet{candes2009exact}, their results are either too loose, or require assuming a particular distribution over the entries. We leverage the more recent results from random matrix theory to derive more informative bounds \cite{bandeira2016sharp,tropp2015introduction}. The proofs are deferred to the appendix. 

The critical implication of Theorem \ref{thm:transductive} is that MCF is the more effective method for transductive CF. Even though the overall asymptotic rate is still controlled by the slack term $\Ocal\Big( \sqrt{\log \frac{1}{\delta} \big / |\Dcal_{\text{test}}|} \Big)$ for both MCF and NCF, we observe a much faster linear decay with respect to $|\Dcal_{\text{train}}|$ in the model complexity term (the second term on RHS of the generalization bounds) of MCF. Since we assume the parameter norms are controlled, it means the advantage of MCF will be more significant as we collect more observations for the same set of users and items. The superiority of MCF for transductive CF is also justified by the experiments, where it outperforms NCF in terms of all the ranking metrics as the gradient descent converges to zero training error (Figure \ref{fig:transductive}). Our result also explains some of the previous success of MCF in such as movie recommendation where the task resembles transductive CF.   

\begin{figure}[hbt]
    \centering
    \includegraphics[width=\linewidth]{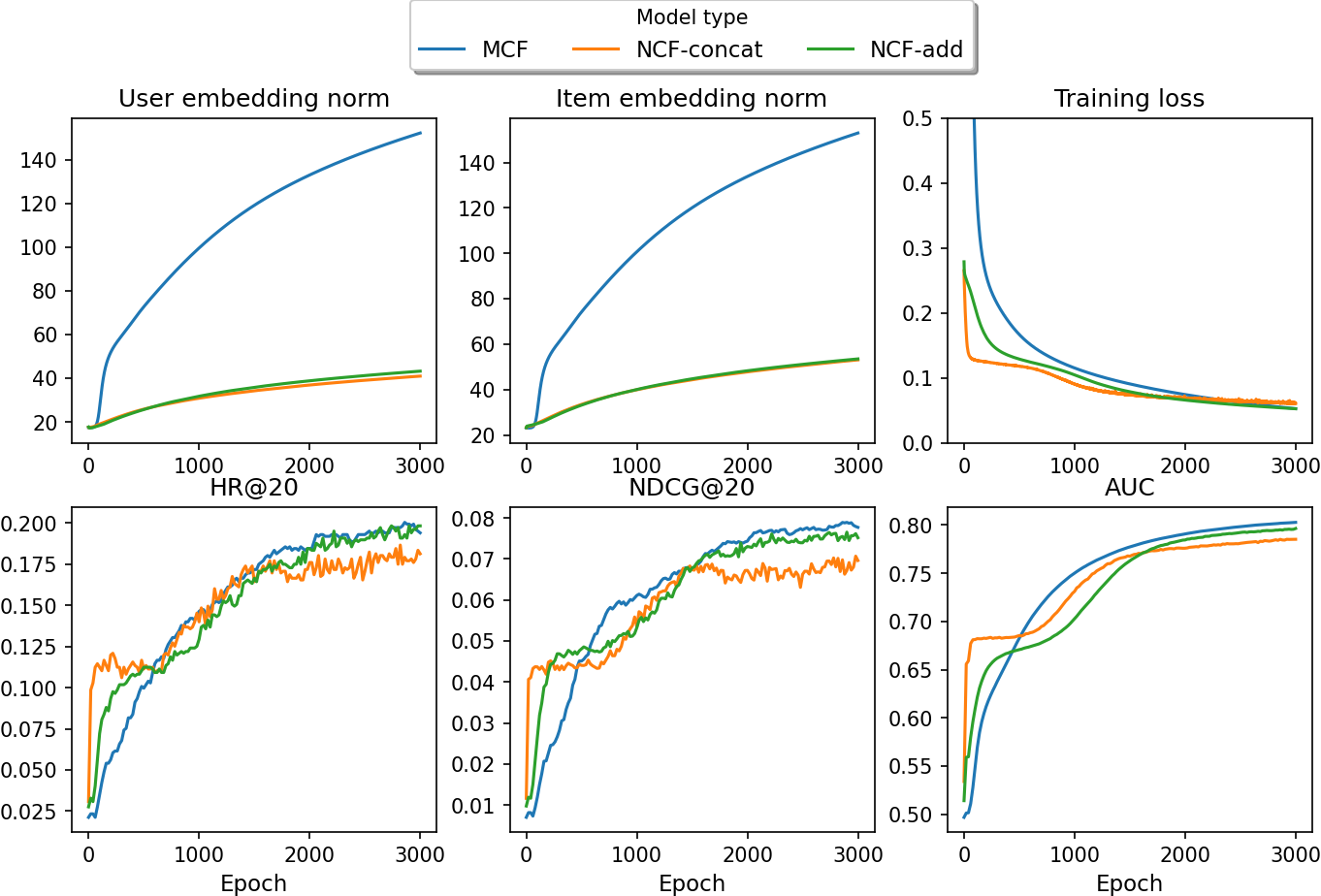}
    \caption{Analysis and model performances for transductive CF using the MovieLens data. We use: $d=32$, $d_1=16$, initializations of $N(0,0.1)$ and learning rate as $0.1$. All the ranking metrics are computed on $\Dcal_{\text{test}}$. The perturbations in the HR and NDCG plots are caused by averaging over ten repetitions. The details are provided in the appendix.}
    \label{fig:transductive}
\end{figure}

The generalization of inductive CF, on the other hand, requires taking account of the discrepancy between the hypothetical $P_\text{unif}$ for the uniform exposure scenario, and the actual exposure distribution $P_O$. As we discussed earlier, using the samples from $P_\text{unif}$ for testing essentially means we want the exposure to be made uniformly at random, so each $(u,i)$ pair has an equal probability of being exposed. It eliminates the exposure bias during testing (evaluation), and is realistic for the modern recommender where the testing samples are often drawn randomly. Therefore, the empirical training risk needs to be adjusted such that it unbiasedly estimate the desired testing risk under $P_{\text{unif}}$. Otherwise, there is no guarantee for generalization. Towards that end, we employ the reweighting method introduced in Section \ref{sec:NCF-MCF} that effectively corrects the exposure bias of the training data. The weight for each $(u,i)$ pair is therefore given by $P_\text{unif}(u,i)/P(O_{u,i}=1) = 1/P(O_{u,i}=1)$, and the testing risk is given by: $\Ebb_{(u,i)\sim P_{\text{unif}}}1\big[ y_{u,i} f(u,i) \leq 0 \big]$, which we denote by the shorthand: $\text{Err}_{P_{\text{unif}}}(f)$. We show the generalization bounds for MCF and NCF in the following two theorems. For brevity, we let $n=|\Dcal_{\text{train}}|$.


\begin{theorem}[\textbf{Inductive NCF}]
\label{thm:inductive-NCF}
Assume the same setting from Theorem \ref{thm:transductive}. Let $\Dcal_{\text{train}}$ be the training data drawn i.i.d according to $P_{\text{unif}}$. Consider the Peasrson $\chi^2$-divergence of two distributions $P$ and $Q$: $D_2(P \| Q) = \int \big(\big(\frac{dP}{dQ}\big)^2 - 1 \big) dQ$. Then with probability at least $1-\delta$, for all $\gamma>0$:
\begin{equation*}
\begin{split}
    & \text{Err}_{P_{\text{unif}}}(f^{\text{NCF}}) \leq \sum_{(u,i) \in \Dcal_{\text{train}}} \frac{1}{P_O(u,i)} 1\big[ y_{u,i} f^{\text{NCF}}(u,i) \leq \gamma \big] \\
    & \quad + \Ocal\bigg( \frac{\sqrt{q}B_{\text{NCF}} \prod_{i=1}^q \lambda_i \cdot \sqrt{D_2(P_{\text{unif}} \,\|\, P_O) + 1}}{\gamma n^{1/2}} \bigg) + \varepsilon,
\end{split}
\end{equation*}
where $\varepsilon = \sqrt{\frac{\log\log_2(4\sqrt{2}/\gamma)}{n}} + \sqrt{\frac{\log(1/\delta)}{n}}$.
\end{theorem}

\begin{theorem}[\textbf{Inductive MCF}]
\label{thm:inductive-MCF}
Following Theorem \ref{thm:inductive-NCF}, we now consider another divergence between $P$ and $Q$: $D_1(P\|Q) =  \int \big(\frac{dP}{dQ} \big) dP$. For all $\gamma>0$, it holds with probability at least $1-\delta$ that:
\begin{equation*}
\begin{split}
    & \text{Err}_{P_{\text{unif}}}(f^{\text{MCF}}) \leq \sum_{(u,i) \in \Dcal_{\text{train}}} \frac{1}{P_O(u,i)} 1\big[ y_{u,i} f^{\text{MCF}}(u,i) \leq \gamma \big] \\
    & \, + \Ocal\bigg( \frac{\sqrt{D_1(P_{\text{unif}} \,\|\, P_O) \big((|\Ucal| + |\Ical|)\lambda_{\text{nuc}} \big)\log 9n}}{\gamma n^{1/2}} \bigg) + \varepsilon,
\end{split}
\end{equation*}
where $\varepsilon = \Ocal\Big(\frac{\log(1/\delta) + (|\Ucal| + |\Ical|)\lambda_{\text{nuc}}\log 9n}{n}\Big)$.
\end{theorem}

The generalization results for correcting distribution shift via reweighting are novel and of interest beyond this paper. The proof of Theorem \ref{thm:inductive-NCF} resembles what we show in our concurrent work of \citet{Xu2021importance}, and we employ the classical "double sampling" trick and a covering number argument to prove Theorem \ref{thm:inductive-MCF}. Note that we use different types of divergence in the bounds of NCF and MCF, which is a matter of technical necessity. Nevertheless, they lead to the same conclusions as we discuss below.

\begin{remark}
The bounds in Theorem \ref{thm:transductive} and \ref{thm:inductive-NCF} are tight up to constants, and the $\log n$ rate in Theorem \ref{thm:inductive-MCF} might be improved, which we discuss in detail in the appendix. The takeaway is that those bounds provide a sound theoretical basis for comparing the models' generalization guarantees. However, we point out that the bounds are not comparable across transductive and inductive setting because the meanings of generalization are different. It is possible to obtain post-hoc generalization bounds, i.e. plug in the trained parameters' norms, using the technique in \citet{koltchinskii2002empirical}, which is beyond the scope of this paper.
\end{remark}

We first observe from Theorem \ref{thm:inductive-NCF} and \ref{thm:inductive-MCF} that in the inductive setting, the rates of generalization are different for NCF and MCF in the asymptotic regime: the NCF has a more favorable $\sqrt{\frac{1}{n}}$ rate. Crucially, in the more interested finite-sample regime, MCF still has an explicit dependency on $|\Ucal|$ and $|\Ical|$ even when $\lambda_{\text{nuc}}$ is controlled. Therefore, regardless of the rate, MCF's model complexity still exceeds NCF by the factor of least $\sqrt{|I|}$, which is significant for modern applications. Since MCF and NCF can achieve a similar empirical loss term (the first term on RHS) as they are both overparameterized, NCF will provide better generalization performance in the inductive setting.

Furthermore, both Theorem \ref{thm:inductive-NCF} and Theorem \ref{thm:inductive-MCF} reveals the divergence between $P_{\text{unif}}$ and $P_O$ as an amplifying factor of the model complexity. In other words, the model with a higher complexity will be more difficult to generalize when the target distribution $P_{\text{unif}}$ is further apart from $P_O$. It also implies that the generalization performance is sensitive to the exposure, emphasising the importance of using the correct $P_O$ for evaluation. If the exposure is misspecified by such as using $P_{n}$, i.e. not correcting for the exposure bias at all, the evaluation results may be misleading. In what follows, we carry out semi-synthetic experiments by controlling the exposure mechanism to illustrate our conclusions.

\begin{figure}[htb]
    \centering
    \includegraphics[width=\columnwidth]{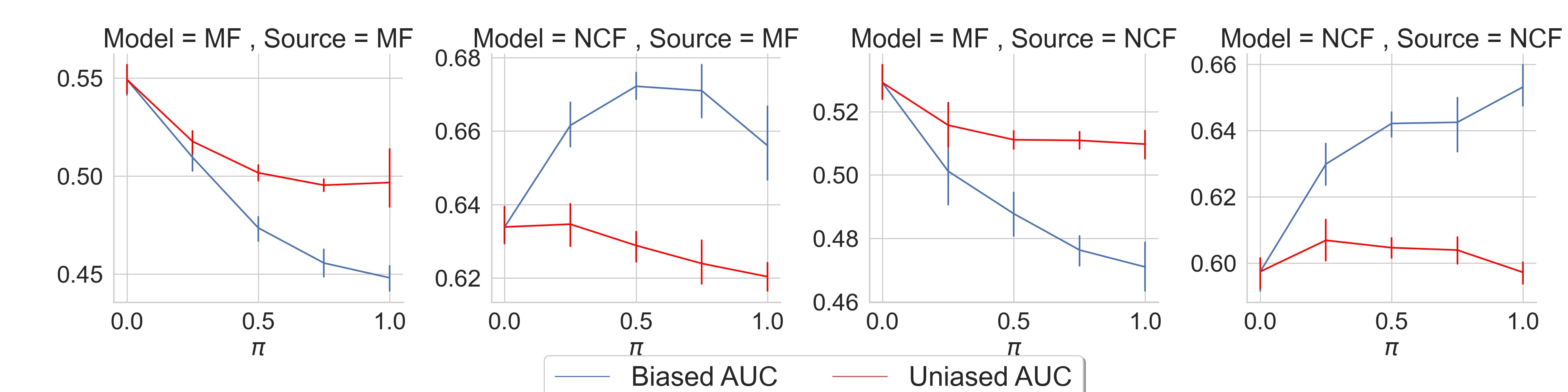}
    \caption{The ranking AUC on the testing data when evaluated using the correct exposure (unbiased evaluation), and treating $P_{\text{unif}}$ as the exposure (biased evaluation). We experiment with $\pi \in \{0,0.25, 0.5, 0.75, 1\}$. In the plot titles, the \emph{source} indicates the model type of $g_{\text{exp}}$ and $g_{\text{rel}}$, and the \emph{model} indicates which CF model we use. The patterns of the unbiased AUC justifies our theoretical argument that it is more difficult to generalize for both NCF and MCF when $\pi$ gets larger, i.e. the discrepancy increases. The experiment details are provided in appendix.}
    \label{fig:inductive}
\end{figure}





To preserve the inductive bias of the MovieLens dataset when generating synthetic data, we use the user and item embeddings learnt from the original dataset for the designed exposure and relevance model. In particular, we generate the click data according to : $p(Y_{u,i}=1) = p(R_{u,i}=1)\cdot p(O_{u,i}=1)$,
where $R_{u,i}$ indicates the relevance. Since we have access to the explicit rating scores from the MovieLens dataset, we first train the relevance model $g_{\text{rel}}(u,i)$ as a regression task. We then train the exposure model $g_{\text{exp}}(u,i)$ by treating each rating event as an exposure, and let the designed exposure model be given by: $P_{\text{design}}: p(O_{u,i}=1)=\sigma\big(g_{\text{exp}}(u,i)\big)$.
Finally, we tune the relevance distribution via: $p(R_{u,i}=1) = \sigma\big(g_{\text{rel}}(u,i) - \mu\big)^{\rho}$, and the goal is to finding the $\mu$ and $\rho$ such that the populational  relevance matches that of the original dataset. We end up with $\mu=3$, $\rho=2$.

For fair comparisons, we experiment using both $f^{\text{MCF}}$ and $f^{\text{NCF}}$ for training $g_{\text{rel}}$ and $g_{\text{exp}}$ so the generating mechanism would not bias toward one of the CF methods. We introduce the hyper-parameter $\pi \in [0,1]$ to control the deviation from $P_O$ to $P_{\text{unif}}$ by using the mixture of: 
\[
P_O = \pi P_{\text{design}} + (1-\pi) P_{\text{unif}}
\]
The idea is to use $\pi$ as a proxy to examine the outcome's sensitivity to $D(P_{\text{unif}}\, \|\, P_O)$ and the degree of misspecification for the exposure model. The various evaluation results are provided in Figure \ref{fig:inductive}. We see that in the biased evaluation where the reweighting is not used, the testing performances drastically change as we increase $\pi$. We point out that this scenario often happens in practice where the exposure bias is not corrected during testing. The relative lift of NCF to MCF, which is shown in Figure \ref{fig:lift}, also suffers from significant perturbations as we increase $\pi$. In contrast, when using the unbiased estimation under the correct reweighting, both the testing metrics and model comparisons are much more consistent, as we show in both Figure \ref{fig:inductive} and \ref{fig:lift}. We conclude from the results that:

\textbf{1.} in the inductive CF setting that we consider, NCF outperforms MCF as suggested by Theorem \ref{thm:inductive-NCF} and \ref{thm:inductive-MCF};

\textbf{2.} when $P_O$ is misspecified by such as$P_{\text{unif}}$, the biased evaluation is very sensitive to the their divergence so the model comparison results can be misleading. 


In this section, we rigorous show for MCF and NCF their generalization behaviors for the transductive and inductive CF. Our results suggest MCF is favorable to transductive CF, while NCF exhibits better guarantees for inductive CF. We further reveal exposure as a critical factor for correctly evaluating and comparing models in the inductive setting. All the conclusions are justified by the numerical results.

\begin{figure}[hbt]
    \centering
    \includegraphics[width=\columnwidth]{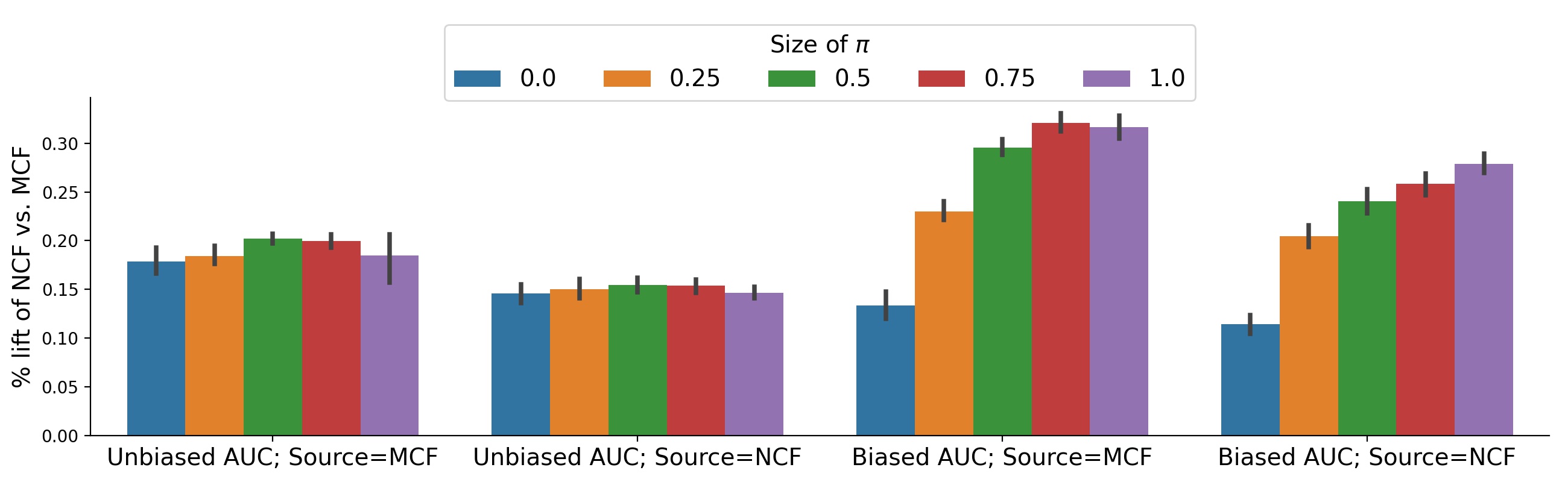}
    \caption{Following Figure \ref{fig:inductive}, we compute the relative percentage lift of the NCF model over MCF, with respect to the unbiased and biased ranking AUC. In the xlabel, we specify the evaluation type, and use \emph{source} to indicate the model type of $g_{\text{exp}}$ and $g_{\text{rel}}$. The lift of NCF w.r.t MCF is much more consistent under the unbiased evaluation.}
    \label{fig:lift}
\end{figure}


\section{Summary, Scope and Limitation}
\label{sec:summary}

We revisit the critical comparison between MCF and NCF by studying the underlying mechanisms of model expressivity, optimization and generalization. We show that the comparisons are better characterized by the optimization paths rather than model expressivity, as they lead to specific norm-constraint solutions that exhibit meaningful generalization patterns under different learning settings. We further emphasize the crucial role of exposure for model evaluation in inductive CF. The fact that they receive less attention in the relevant literature might explain some of the controversial outcomes.

In the concurrent work of \citet{xu2021theoretical}, the authors primarily study the theoretical perspectives of user and item embeddings learnt by the skip-gram negative sampling algorithm, which deviates from the scope of regular CF. On the other hand, some of our analysis is carried out under conditions C1-C3, which may restrict our result to generalize to the more sophisticated NCF architectures. Also, we do not study how the various tricks of training and data manipulation could affect our conclusions.
Finally, while our results strongly advocate for unbiased evaluation in the inductive setting, it is nevertheless very challenging to implement the inverse weighting approach in practice. It is pointed out by \citet{xu2020adversarial} that there often exists identifiability issues, and more importantly, the inverse propensity weighting is not applicable when the exposure is deterministic. We leave to future work to study how to leverage reweighting for recommender systems more effectively.

\bibliography{references}
\bibliographystyle{icml2021}

\onecolumn
\appendix
\setcounter{equation}{0}
\renewcommand{\theequation}{A.\arabic{equation}}
\setcounter{figure}{0}
\renewcommand{\thefigure}{A.\arabic{figure}}
\setcounter{table}{0}
\renewcommand{\thetable}{A.\arabic{figure}}
\setcounter{section}{0}
\renewcommand{\thesection}{A.\arabic{section}}
\setcounter{lemma}{0}
\renewcommand{\thelemma}{A.\arabic{lemma}}
\setcounter{corollary}{0}
\renewcommand{\thecorollary}{A.\arabic{corollary}}

We provide the technical proofs, experiment details as well as the relegated discussions mentioned in the paper. The appendix for Section 4, 5, 6 are provided in \ref{sec:sec4}, \ref{sec:sec5}, \ref{sec:sec6}, respectively. The auxiliary lemmas in our proofs are summarized in \ref{sec:lemma}. The additional experiment details are provided in \ref{sec:add-experiment}.

\section{Material for Section 4}
\label{sec:sec4}

When optimized by the gradient descent: $\thetabf^{(t)} = \thetabf^{(t-1)} - \eta \nabla \Lcal(\thetabf^{(t-1)})$ using an infinitesimal learning rate, the updates in the parameter space can be equivalently described by the \emph{gradient flow}:
\[
\frac{d \thetabf^{(t)}}{dt} = -\nabla_{\thetabf} \Lcal\big(\thetabf^{(t)} \big).
\]

A nice property of gradient flow is that if $\Lcal$ is smooth, then the objective function is non-increasing during the updates since:
\begin{equation}
\label{eqn:grad-flow}
\frac{d\Lcal\big(\thetabf^{(t)} \big)}{dt} = -\Big \langle \nabla \Lcal\big(\thetabf^{(t)} \big), \frac{d \thetabf^{(t)}}{dt} \Big \rangle = -\Big\| \frac{d \thetabf^{(t)}}{dt} \Big\|_2^2,
\end{equation}
which is non-negative. Therefore, it saves the discussion of choosing the proper learning rate to ensure the same property in gradient descent. 

Another preparation work is to reformulated NCF, especially the input (which are essentially the users and items embeddings), into a standard form of FFN: $\Wbf_1 \sigma\big(\Wbf_2 \xbf_{u,i} \big)$ where $\xbf_{u,i}$, are fixed and do not depend on the unknown embeddings. 

We use $\ebf_u \in \Rbb^{|\Ucal|}$ and $\ebf_i\in\Rbb^{|\Ical|}$ to denote the one-hot encoding of the user and item id. Also, we use $\ebf_{u,i} \in \Rbb^{|\Ucal||\Ical|}$ to denote the one-hot encoding of user+item id combined. Therefore, NCF with addition can be efficiently represented as:
\begin{equation}
\label{eqn:ncf-a}
f^{\text{NCF-a}}(u,i) = \Wbf_1\sigma\big(\Wbf_2 \xbf_{u,i}  \big), \text{ with } \Wbf_2 = \big[\Zbf_U^{\intercal}, \Zbf_I^{\intercal}\big]^{\intercal} \text{ and } \xbf_{u,i} = [\ebf_u, \ebf_i]^{\intercal};
\end{equation}
and NCF with concatenation is:
\begin{equation}
\label{eqn:ncf-c}
f^{\text{NCF-c}}(u,i) = \Wbf_1\sigma\big(\Wbf_2 \xbf_{u,i}  \big), \text{ with } \Wbf_2 = \begin{bmatrix}
\zbf_{i_1} & \zbf_{u_1} \\
\zbf_{i_1} & \zbf_{u_2} \\
\dots & \cdots 
\end{bmatrix}^{\intercal} \text{ and } \xbf_{u,i} = \ebf_{u,i}^{\intercal}.
\end{equation}

Recall the linearization from Section 4, where we denote the first-order Taylor approximation of $f(\thetabf;\, \cdot)$ by $\tilde{f}(\thetabf;\, \cdot)$ such that:
\[
\tilde{f}\big(\thetabf; (u,i)\big) := f\big(\thetabf^{(0)};(u,i)\big) + \Big \langle \thetabf - \thetabf^{(0)} , \nabla f\big(\thetabf^{(0)};(u,i)\big) \Big \rangle.
\]
Also, we use $d$, $d_1$ to denote the embedding dimension and the dimension of the first hidden layer in the FFN for NCF, and assume $d_1=d$ w.l.o.g. We still consider the scaled initialization $N(0, \alpha/d)$ where $\alpha$ is a constant.
Under the infinite-width limit, we can show that the NTK converges to a fixed kernel at initialization, which we referred to as the collaborative filtering kernel. 

\begin{lemma}
\label{lemma:cf-kernel}
For the MCF and NCF we described in Section 3.1, the neural tangent kernel $K\big( (u,i), (u',i')  \big) = \big\langle \nabla f\big(\thetabf; (u,i)\big) , \nabla f\big(\thetabf; (u',i')\big) \big\rangle$ have the following convergence result:
\[
\lim_{d\to\infty} K\big( (u,i), (u',i')  \big) = K_{\text{CF}}\big((u,i) , (u',i') \big) := a + b\cdot 1[i=i'] + c \cdot 1[u=u'],
\]
where under the $N(0,1/d)$ initializations, $a=0$, $b=c=1$ for $K_{\text{MCF}}$; $a=1/\pi$, $b=c=\frac{1}{2} - \frac{1}{2\pi}$ for $K_{\text{NCF-c}}$, and $a=1/\pi$, $b=c=\frac{1}{2} - \frac{(2-\sqrt{3})}{2\pi}$ for $K_{\text{NCF-a}}$.
\end{lemma}

\begin{proof}
We first consider NCF. We first reformulate NCF's formulation in (\ref{eqn:ncf-a}) and (\ref{eqn:ncf-c}) as:
\[ 
f(\xbf_{u,i}) = \sqrt{\frac{2}{d}}\Wbf_1\sigma\big(\Wbf_2 \xbf_{u,i} \big), \Wbf_1, \Wbf_2 \sim N(0,1),
\]
where we extract the $1/d$ variance to the front, and add the $\sqrt{2}$ factor for convenience. 

Notice that: $\frac{\partial f(\xbf_{u,i})}{\partial \Wbf_{1,j}} = \sqrt{2/d} \big( \Wbf_{2,j}^{\intercal} \xbf_{u,i}\big)_{+}$, and 
\[
\nabla_{\Wbf_{2,j}} f(\xbf_{u,i}) = \sqrt{2/d} \Wbf_{1,j} \xbf_{u,i} 1\big[ \Wbf_{2,j}^{\intercal} \xbf_{u,i} \geq 0 \big],
\] 
where $\Wbf_{1,j}$ is the $j^{th}$ element of the vector $\Wbf_1$, and $\Wbf_{2,j}$ is the $j^{th}$ column of the matrix $\Wbf_2$. For notation simplicity, we define $v_j = \Wbf_{1,j}$ and $\wbf_j = \Wbf_{2,j}$.

Therefore, the NTK for NCF is given by:
\begin{equation}
\label{eqn:ntk-ncf}
\begin{split}
    &\nabla f(\xbf_{u,i})^{\intercal} \nabla f(\xbf_{u',i'})  = \\ & \quad \frac{2}{d}\sum_{j=1}^d \big(\wbf_j^{\intercal}\xbf_{u,i}\big)_{+} \big(\wbf_j^{\intercal}\xbf_{u',i'}\big)_{+} + \frac{2}{d}\sum_{j=1}^d(v_j\xbf_{u,i})^{\intercal}(v_j\xbf_{u',i'}) 1\big[\wbf_j^{\intercal}\xbf_{u,i}\geq 0 \big] 1\big[\wbf_j^{\intercal}\xbf_{u',i'}\geq 0 \big].
\end{split}
\end{equation}

Using the mean and variance formula of truncated normal distribution, following the setup in (\ref{eqn:ncf-a}) and (\ref{eqn:ncf-c}), for NCF with concatenation we have:
\begin{itemize}
    \item When $u\neq u'$ and $i\neq i'$, we have:
    \begin{equation*}
    \begin{split}
        K\big((u,i), (u',i')\big) &= \frac{2}{d}\sum_{j=1}^d (\wbf_j)_{+} (\wbf^*_j)_{+}, \quad \wbf^*_j \text{ is an i.i.d copy of } \wbf_j \\
        & \overset{d\to\infty}{=} 2\Ebb\big[ (\wbf_j)_{+} (\wbf^*_j)_{+} \big] = \frac{4}{\pi}
    \end{split}
    \end{equation*}
    \item When $u=u'$ or $i=i'$, we have: 
    \begin{equation*}
    \begin{split}
        K\big((u,i), (u',i')\big) &= \frac{1}{d}\sum_{j=1}^{d} (\wbf_j)_{+}^2 + \frac{2}{d}\sum_{j=1}^d (\wbf_j)_{+} (\wbf^*_j)_{+}, \quad \wbf^*_j \text{ is a copy of } \wbf_j  \\
        & \overset{d\to\infty}{=} var\big((\wbf_j)_{+}\big) + 2\Ebb\big[ (\wbf_j)_{+} (\wbf^*_j)_{+} \big] = 2 + \frac{2}{\pi}
    \end{split}
    \end{equation*}
    \item When $u=u'$ and $i=i'$, we leverage the integral formulation of arc-cosine kernel $K_0$ in Lemma \ref{lemma:arc-cos-kernel} such that:
    \begin{equation*}
    \begin{split}
        K\big((u,i), (u',i')\big) & = K_{0}(\xbf_{u,i}, \xbf_{u',i'}) -  \frac{2}{d}\sum_{j=1}^d (\wbf_j)_{+} (\wbf^*_j)_{+} + \frac{2}{d}\sum_{j=1}^{d} (\wbf_j)_{+}^2  \\
        & \overset{d\to\infty}{=} 8 - \frac{16}{\pi}.
    \end{split}
    \end{equation*}
\end{itemize}

The results for NCF under addition is obtained using basically the same computations. 
For MCF, on the other hand, we reformulate the predictor as: $f(u,i) = \frac{1}{d}\big\langle \Zbf_U\Zbf_I, \Xbf_{u,i}  \big \rangle$ where $\Xbf_{u,i} = \ebf_u \ebf_i^{\intercal}$, and the embeddings follow $N(0,1)$ initializations. Then it holds that:
\begin{equation}
\label{eqn:ntk-mcf}
    \nabla f(\Xbf_{u,i})^{\intercal} \nabla f(\Xbf_{u',i'}) = \frac{1}{d} \Big( \langle \zbf_u,\zbf_{u'} \rangle 1[i=i'] + \langle \zbf_i,\zbf_{i'} \rangle 1[u=u'] \Big),
\end{equation}

which directly leads to the stated results of $K_{\text{CF}}$.

\end{proof}

\begin{remark}[The parameterization of $K_{\text{CF}}$ and the intializations of MCF, NCF]
It is evident from the above proof that the relative scale of $a$, $b$ and $c$ in $K_{\text{CF}}$ can depend on the constant term $\alpha$ in the  intializations of $N(0,\alpha/d)$. For instance, if the infinite-width MCF initializes the user embeddings $\Zbf_U$ with $N(0,\alpha_1/d)$ and the item embeddings $\Zbf_I$ with $N(0,\alpha_2/d)$, then by (\ref{eqn:ntk-mcf}) we immediately have $b=\alpha_1$ and $c=\alpha_2$. 

Also, the parameterization of $K_{\text{CF}}$ for NCF is also dependent on the initialization. We observe from (\ref{eqn:ntk-ncf}) that $a$ would not change as long as the initializations are i.i.d., but $b$ and $c$ again depend on the individual $\alpha$. The exact derivations for the NTK of FFN is stuided by \cite{arora2019exact}, and in \cite{yang2019scaling} the author provides the NTK formulation for a broad range of neural networks.
\end{remark}

Other than the convergence to a fixed kernel at initialization, the infinite-width limit also suggests that the parameters varies little during the gradient flow updates, and the linearization of $\tilde{f}$ has a good approximation:
\[
\frac{\big\|\thetabf^{(t)} - \thetabf^{(0)}\big\|_2}{\big\|\thetabf^{(0)}\big\|_2} = \Ocal(\sqrt{1/d}), \, \text{ and } \tilde{f}\big(\thetabf^{(t)}; (u,i)\big) = f\big(\thetabf^{(t)};(u,i)\big) + \Ocal(\sqrt{1/d}).
\]
We formalize the above arguments in the following lemma.

\begin{lemma}
\label{lemma:response-surface}
Let the gradient flow updates under $\tilde{f}$ be denoted by $\tilde{\thetabf}^{(t)}$. Under the exponential loss or log loss, when the predictor $f(\thetabf; \cdot)$ is local Lipschitz and admits chain rule, the corresponding decision boundaries for the two gradient flow trajectories satisfy the following result for any $T>0$:
\begin{equation}
\label{eqn:grad-flow-converge}
\lim_{d\to\infty} \sup_{t\leq T} \big\|F(\thetabf^{(t)}) - \tilde{F}(\tilde{\thetabf}^{(t)})\big\|_2 = 0.
\end{equation}
\end{lemma}

To the best of our knowledge, the similar infinite-width convergence results were studied under the squared loss \cite{arora2019exact,jacot2018neural}, and we extend them to the classification setting.

\begin{proof}
By the chain rule, for any step $T>0$, we have:
\[
\int_0^T \big\|\thetabf^{(t)}\big\|_2 dt = \int_0^T \big\|\nabla\Lcal(\thetabf^{(t)})\big\|_2 dt \leq \sqrt{T} \Big( \int_0^T \big\|\nabla\Lcal(\thetabf^{(t)})\big\|_2^2 dt\Big)^{1/2},
\]
by the H\"{o}lder's inequality. According to (\ref{eqn:grad-flow}), $d\Lcal\big(\thetabf^{(t)} \big) / dt = -\Big\| \nabla \Lcal(\thetabf^{(t)})\Big\|_2^2$, so we have:
\[
\sup_{t\leq T} \big\|\thetabf^{(t)} - \thetabf^{(0)}\big\|_2 \leq \sqrt{T \Lcal\big(\thetabf^{(t)} \big))} \lesssim \sqrt{1/d},
\]
due to the scaled initializations. Denote the risk associated with a predictor by $R\big(f(\thetabf)\big)$. It can then be deduced that $\sup_{t\leq T} \big\|\thetabf^{(t)} - \thetabf^{(0)}\big\|_2 \leq C_1$ and $\sup_{t\leq T} \big\|\nabla R\big(f(\thetabf^{(t)})\big)\big\|_2 \leq C_2$, for some constants $C_1$ and $C_2$.

Define $\delta(t) := \|f(\thetabf^{(t)}) - \tilde{f}(\thetabf^{(t)})\|_2$, and the NTK under a specific $\thetabf$ as $K(\thetabf)$. Since $\|K(\thetabf)^{(t)} - K(\thetabf)^{(0)}\|_2 = \mathcal{O}(\sqrt{1/d})$ according to \cite{jacot2018neural,arora2019exact}, we have:
\begin{equation*}
\begin{split}
\frac{d\delta(t)}{dt} & \leq \Big\| K(\thetabf^{(t)}) \nabla R\big(f(\thetabf^{(t)})\big) - K(\thetabf^{(0)}) R\big(\tilde{f}(\thetabf^{(t)})\big)\Big\|_2  \\ 
& \leq \Big\|\big(K(\thetabf^{(t)})-K(\thetabf^{(0)})\big) \nabla R\big(f(\thetabf^{(t)}) \Big\|_2  + \Big\|K(\thetabf^{(0)})\big(R\big(f(\thetabf^{(t)})-R\big(\tilde{f}(\thetabf^{(t)})\big) \Big\|_2  \\
& \leq C_3/\sqrt{d} + C_4 \delta(t),
\end{split}
\end{equation*}
for some constants $C_3$ and $C_4$. Since $\delta(0)=0$, then $\delta_t$ is a sub-solution to the ordinary differential equation: $\frac{d\delta(t)}{dt}=C_3/\sqrt{d} + C_4 \delta(t)$ with $\delta(0)=0$. It then follows: $\delta(t) \leq \frac{C_3(\exp(C_4 t)-1)}{\sqrt{d} C_4}$, so we conclude that: $\lim_{d \to \infty} \sup_{t\leq T} \|F(\thetabf^{(t)}) - \tilde{F}(\tilde{\thetabf}^{(t)})\|_2 = 0$.
\end{proof}

\textbf{Proof for Theorem 1}
\begin{proof}
According to Corollary \ref{corollary:max-margin} from \ref{sec:sec5}, under the exponential or log loss, the linearization $\tilde{f}(\tilde{\thetabf}; \, \cdot)$ satisfies condition \textbf{C1}, \textbf{C2} and \textbf{C3}, so the gradient flow optimization of $\tilde{\thetabf}^{(t)}$ converges to the stationary points of:
\[
\min \, \|\tilde{\thetabf}\|_2 \, \text{ s.t. } y_{u,i} \tilde{f}(\tilde{\thetabf};(u,i)) \geq 1, \, \forall (u,i)\in\Dcal_{\text{train}}.
\]
Combining the results from Lemma \ref{lemma:cf-kernel} and Lemma \ref{lemma:response-surface}, under the gradient flow optimization, the response surface of MCF and NCF converges to the minimum RKHS norm solution: 
\[
\lim_{t\to\infty} \lim_{d\to\infty} F\Big(\frac{\thetabf^{(t)}}{\|\thetabf^{(t)}\|_2}\Big) \overset{\text{stationary points of}}{\to}  \argmin_{f:(\Ucal, \Ical)\to \Rbb} \big\|f\big\|_{K_{\text{CF}}} \, \text{s.t.} \, y_{u,i}f(u,i) \geq 1, \, \forall (u,i)\in \Dcal_{\text{train}}.
\]
\end{proof}

\section{Material for Section 5}
\label{sec:sec5}

We first discuss the implications of condition \textbf{C1}, \textbf{C2} and \textbf{C3}. Recall that: 

\textbf{C1}. The loss function has the exponential-tail behavior such as the exponential loss and log loss;

\textbf{C2}. Both the MCF and NCF in are $L$-homogeneous, i.e. $f(\thetabf;\,\cdot) = \|\thetabf\|_2^L \cdot f\big(\thetabf/\|\thetabf\|_2;\,\cdot\big)$ for some $L>0$, and have some smoothness property;

\textbf{C3}. The data is separable with respect to the overparameterized MCF and NCF.

The exponential decay on the tail of the loss function is important for the inductive bias of gradient descent as $\ell(u)$ behaves like $\exp(-u)$ when $u\to \infty$. \citet{soudry2018implicit} first propose the notion of \emph{tight exponential tail}, where the negative loss derivative $-\ell'(u)$ behave like:
\[ 
-\ell'(u) \lesssim \big(1+\exp(-c_1 u)\big)e^{-u} \text{ and } -\ell'(u) \gtrsim \big(1-\exp(-c_2 u)\big)e^{-u},
\]
for sufficiently large $u$, where $c_1$ and $c_2$ are positive constants. There is also a smoothness assumption on $\ell(\cdot)$. There is a more general (and perhaps more direct) definition of exponential-tail loss function \cite{lyu2019gradient}, where $\ell(u) = \exp(-f(u))$, such that:
\begin{itemize}
    \item $f$ is smooth and $f'(u) \geq 0, \forall u$;
    \item there exists $c>0$ such that $f'(u)u$ is non-decreasing for $u > c$ and $f'(u)u \to \infty$ as $u \to \infty$.
\end{itemize}
It is easy to verify that the exponential loss, log loss and cross-entropy loss satisfy both requirements.

The predictor of MCF is obviously homogeneous, but for NCF to be homogeneous, the bias terms cannot be used for the hidden layers in the FFN. The requirement on the activation function is relative mild, since ReLU, LeakyReLU and some other common activation functions all preserve the homogeneity of the predictor. 

On the other hand, the smoothness condition, which includes the locally Lipschitz condition and differentiability. 
Notice that Lipschitz condition is rather mild assumption for neural networks, and several recent paper are dedicated to obtaining the Lipschitz constant of deep learning models using activation such as ReLU \citep{fazlyab2019efficient,virmaux2018lipschitz}. 
The differentiablity condition is more technical-driven such that we can analyze the gradients. In practice, neural networks with ReLU activation do not satisfy the condition. We point out that there do exist smooth homogeneous activation functions, such as the quadratic activation $\sigma(x) = x^2$. Nevertheless, the ReLU activation admits the chain rule, so the same analysis using gradients can be carried out by the sub-differentials. Therefore, in our experiments, we use ReLU as the activation function for convenience, and assume differentiablity to provide a more straightforward analysis. 

Finally, by implying the separability, we also assume that there exists $t_0$ such that:
\begin{equation}
\label{eqn:seperable}
\forall t>t_0, \quad \frac{1}{|\Dcal_{\text{train}}|} \sum_{(u,i)\in \Dcal} \exp\Big(-y_{u,i} f\big(\thetabf^{(t)};(u,i)\big)\Big) < 1.
\end{equation}

In the following corollary, we prove a general result for gradient flow converging to the max-margin solution under condition.

\begin{corollary}
\label{corollary:max-margin}
For gradient flow optimization with the exponential loss, under condition  \textbf{C1}, \textbf{C2} and \textbf{C3}, $\lim_{t\to\infty}\thetabf^{(t)}/\|\thetabf^{(t)}\|_2$ converges to the KKT points of:
\begin{equation*}
\min \|\thetabf\|_2 \quad \text{s.t.} \quad y_{u,i}f\big(\thetabf;(u,i)\big) \geq 1, \, \forall (u,i)\in\Dcal_{\text{train}}.
\end{equation*}
\end{corollary}

Compared with the results in \cite{gunasekar2018characterizing}, we do not assume the loss function already converged in direction. Although we use the same type of constraint quality idea as in \cite{lyu2019gradient}, their results focus on the convergence of the normalized margin and our result emphasize the dynamics in the parameter space. The interests in this result is also beyond the content of this paper.

\begin{proof}
First notice that the KKT condition for the original problem (where we add a $\frac{1}{2}$ factor for convenience): $\min \frac{1}{2}\|\thetabf\|_2 \, \text{s.t.} \, y_{u,i}f\big(\thetabf;(u,i)\big) \geq 1, \, \forall (u,i)\in\Dcal_{\text{train}}$ is given by:
\begin{equation}
\label{eqn:KKT}
\exists \lambda_{u,i}\geq 0, (u,i)\in \Dcal_{\text{train}}, \text{ s.t. } \left\{\begin{array}{l}
\thetabf + \sum_{(u,i)\in\Dcal_{\text{train}}}\lambda_{u,i}y_{u,i}\nabla f\big(\thetabf;(u,i)\big) = 0 \\
\lambda_{u,i} \Big(y_{u,i}f\big(\thetabf;(u,i)\big) - 1 \Big) = 0, \, \forall (u,i)\in\Dcal_{\text{train}}.
\end{array}
\right.
\end{equation}
As we mentioned in Section 5, when the constraints are non-convex, the stationarity of local or global optimum does not equal KKT optimality, but when the Guignard constraint qualification (GCQ) is satisfied \cite{gould1971necessary}, those two become exchangeable. GCQ might be the weakest constraint qualification in some sense, but it is very difficult to check in practice. 

On the other hand, the Mangasarian-Fromovitz constraint qualification (MFCQ), though stronger than GCQ (and thus imply GCQ), is easier to examine. Specifically, it states for the stationary points that:
\[
\exists\, \dbf \quad \text{ s.t. } \quad \Big\langle y_{u,i}\nabla f\big(\thetabf;(u,i)\big),\dbf \Big\rangle > 0, \, \text{for all }(u,i)\in\Dcal_{\text{train}}.
\]
Notice that $L$-homogeneous functions all satisfy MFCQ, because according to Lemma \ref{lemma:homogeneous}, for any stationary point $\thetabf^*$ satisfying $y_{u,i}f\big(\thetabf^*;(u,i)\big)=1$, we let $\dbf^* = \thetabf^*$ and it holds that:
\[
\Big\langle y_{u,i}\nabla f\big(\thetabf^*;(u,i)\big),\dbf^* \Big\rangle = L > 0.
\]
As a consequence, the stationary points for the L-homogeneous predictors are indeed the KKT points. Then we show the convergence of the gradient flow optimization path to the KKT points. We first define the quantity:
\[
\tgamma\big(\thetabft\big) := \frac{-\log \sum_{(u,i)} \exp\big(-y_{u,i}f\big(\thetabft;(u,i)\big)\big)}{\|\thetabf^{(t)}\|_2^2} = \|\thetabft\|_2^2 \cdot \log\frac{1}{\Lcal(\thetabft)},
\]
which is smoothed version of the average margin normalized by the $\|\thetabft\|_2^2$.

We show the convergence by the following three steps.

\textbf{S1}. Under the gradient flow optimization, $\tgamma\big(\thetabft\big)$ is non-decreasing for $t\geq t_0$ (\ref{eqn:seperable}), together with $\Lcal(\thetabft) \to 0$ and $\|\thetabft\|_2 \to \infty$.

\textbf{S2}. With a scaling factor $\alpha>0$, it holds that: $\exists \lambda_{u,i}(t)\geq 0, (u,i)\in \Dcal_{\text{train}}, \text{ s.t. }$
\begin{equation}
\label{eqn:S2}
\left\{\begin{array}{ll}
    \Big\|\alpha \thetabft - \sum_{(u,i)\in\Dcal_{\text{train}}}\lambda_{u,i}y_{u,i}\nabla f\big(\alpha\thetabft;(u,i)\big) \Big\|_2 \lesssim \bigg( 1- \Big\langle\frac{\thetabft}{\|\thetabft\|_2}, \frac{d\thetabft/dt}{\|d\thetabft/dt\|_2} \Big\rangle \bigg)\frac{1}{\tgamma\big(\thetabft\big)} & (\textbf{S2.1}) \\
    \lambda_{u,i}(t) \Big(y_{u,i}f\big(\alpha\thetabft;(u,i)\big) - 1 \Big) \lesssim \frac{1}{\tgamma\big(\thetabft\big) \|\thetabft\|_2^2} &  (\textbf{S2.2})
\end{array}
\right.
\end{equation}
We mention that \textbf{S2.2}$\overset{t\to\infty}{\to} 0$ will be a consequence of \textbf{S1}, and to show \textbf{S2.1}$\overset{t\to\infty}{\to} 0$, we need to prove $\Big\langle\frac{\thetabft}{\|\thetabft\|_2}, \frac{d\thetabft/dt}{\|d\thetabft/dt\|_2} \Big\rangle \to 1$.
Once we have \textbf{S2.1}$\overset{t\to\infty}{\to} 0$ + \textbf{S2.2}$\overset{t\to\infty}{\to} 0$ + MFCQ, it holds that $\lim_{t\to\infty}\thetabft / \|\thetabft\|_2$ are proportional to the KKT points. 

\textbf{S3}. Finally, the goal is show that $\Big\langle\frac{\thetabft}{\|\thetabft\|_2}, \frac{d\thetabft/dt}{\|d\thetabft/dt\|_2} \Big\rangle \to 1$. 

We follow our game plan and first prove the results in \textbf{S1}. To show that $\tgamma\big(\thetabft\big)$ is a non-decreasing function of $t$, we derive $\frac{d\log\tgamma\big(\thetabft\big) }{dt}$ and leverage (\ref{eqn:grad-flow}) to show it's non-negative:
\begin{equation}
\label{eqn:margin-differential}
\begin{split}
    \frac{d\log\tgamma\big(\thetabft\big) }{dt} & = \frac{d\log\log \frac{1}{\Lcal(\thetabft)}}{dt} - L\frac{d\log \|\thetabft\|_2}{dt} \\
    & = \frac{1}{\log\frac{1}{\Lcal(\thetabft)}\Lcal(\thetabft)} \bigg(-\frac{d\Lcal(\thetabft)}{dt} \bigg) - L\cdot\frac{d\log \|\thetabft\|_2}{dt}.
\end{split}
\end{equation}
Define $q_{u,i}(t) := y_{u,i}f(\thetabft;(u,i))$ to be the margin of each data point during optimization. Notice that:
\begin{itemize}
\item For all $(u,i)\in \Dcal_{\text{train}}$, $\min_{u,i} q_{u,i}(t) \geq \log\frac{1}{\Lcal(\thetabft)}$, therefore:
\begin{equation}
    \label{eqn:Qt}
    \log\frac{1}{\Lcal(\thetabft)}\Lcal(\thetabft) \leq \sum_{(u,i)}\exp\big(-q_{u,i}(t) \big)q_{u,i}(t),
\end{equation}
and we denote the RHS by $Q(t)$. By the separability assumption, we have $\Lcal(\thetabft) < 1$ for $t>t_0$, which indicates $Q(t)>0$ for $t>t_0$.
\item It holds that:
\begin{equation*}
\begin{split}
    \frac{d\|\thetabft\|_2}{dt} & = 2\Big\langle \thetabft, \frac{d\thetabft}{dt} \Big\rangle \\
    & = 2\Big\langle \thetabft, \nabla \Lcal\big(\thetabft\big) \Big\rangle \quad \text{by (\ref{eqn:grad-flow})} \\
    & = 2\Big\langle \thetabft, \sum_{(u,i)}\exp\big(-y_{u,i}f\big(\thetabft; (u,i)\big)\big)\cdot y_{u,i}\nabla f\big(\thetabft; (u,i)\big) \Big\rangle \\
    & = 2L\sum_{(u,i)} \exp\big(-q_{u,i}(t) \big)q_{u,i}(t) \quad \text{by Lemma \ref{lemma:homogeneous}.}
\end{split}
\end{equation*}
Hence, we have:
\begin{equation}
\label{eqn:dlog-dt}
\begin{split}
    \frac{d\log\tgamma\big(\thetabft\big) }{dt} & = \frac{1}{2\|\thetabft\|_2^2}\frac{d\|\thetabft\|_2^2}{dt} = \frac{L}{\|\thetabft\|_2^2} \exp\big(-q_{u,i}(t) \big)q_{u,i}(t) \\
    & = \frac{\big\langle \thetabft, d\thetabft / dt \big\rangle}{\|\thetabft\|_2^2}.
\end{split}
\end{equation}
\end{itemize}

Combining (\ref{eqn:grad-flow}), (\ref{eqn:margin-differential}), (\ref{eqn:Qt}) and (\ref{eqn:dlog-dt}):
\begin{equation*}
\begin{split}
\frac{d\log\tgamma\big(\thetabft\big) }{dt} & \geq \frac{1}{Q(t)}\bigg(\Big\|\frac{d\thetabft}{dt} - \Big\langle \frac{\thetabft}{\|\thetabft\|_2^2}, \frac{d\thetabft}{dt} \Big\rangle \Big\|_2^2 \bigg) \\
& = \frac{\|\thetabft\|_2^2}{Q(t)}\bigg\|\frac{d\thetabft / \|\thetabft\|_2^2}{dt} \bigg\|_2^2 = L\bigg(\frac{d\log\|\thetabft\|_2}{dt} \bigg)^{-1} \bigg\|\frac{d\tilde{\thetabft}}{dt} \bigg\|_2^2 \geq 0
\end{split}
\end{equation*}
where we defined $\tilde{\thetabft} = \frac{\thetabft}{\|\thetabft\|_2}$.

Then we show $\Lcal(\thetabft) \to 0$ and $\|\thetabft\|_2\to 0$ using the monotonicity of $\tgamma\big(\thetabft\big)$. Note that:
\[
\frac{-d\Lcal(\thetabft)}{dt} = \Big\|\frac{d\thetabft}{dt} \Big\|_2^2 \geq \Big\langle \frac{\thetabft}{\|\thetabft\|_2}, \frac{d\thetabft}{dt} \Big\rangle^2 = L^2 \frac{Q(t)^2}{\|\thetabft\|_2^2} \quad \text{by (\ref{eqn:dlog-dt})},
\]
where the inequality is by applying the Cauchy-Schwartz inequality. By (\ref{eqn:Qt}): $Q(t)\geq \Lcal(\thetabft)\log\frac{1}{\Lcal(\thetabft)}$, and by the definition, we have: $\|\thetabft\|_2 = \Big(\log\frac{1}{\Lcal(\thetabft)} / \tgamma\big(\thetabft\big) \Big)^{1/L}$. As a consequence,
\begin{equation*}
\begin{split}
    \frac{-d\Lcal(\thetabft)}{dt} & \geq L^2 \tgamma\big(\thetabft\big)^{2/L} \log\frac{1}{\Lcal(\thetabft)}^{2-2/L} \Lcal(\thetabft) \\
    & \geq L^2 \tgamma\big(\thetabf^{(t_0)}\big)^{2/L} \log\frac{1}{\Lcal(\thetabft)}^{2-2/L} \Lcal(\thetabft),
\end{split}
\end{equation*}
which indicates that:
\[
\frac{d\big(1/\Lcal(\thetabf) \big)}{dt} \frac{1}{\log\frac{1}{\Lcal(\thetabft)}^{2-2/L}} \geq l^2 \tgamma\big(\thetabf^{(t_0)}\big)^{2/L}.
\]
Taking the intergral on both sides from $t_0$ to $t$, we immediate have:
\[
\int_{1/\Lcal(\thetabf^{(t_0)})}^{1/\Lcal(\thetabft)} \big(\log u\big)^{2/L-2}du \geq L^2\tgamma\big(\thetabf^{(t_0)}\big)^{2/L}(t-t_0) \overset{t\to\infty}{\to} \infty.
\]
Therefore, $1/\Lcal(\thetabft)$ must diverge when $L\geq 1$, which implies $\Lcal(\thetabft) \to 0$. Due to the L-homogeneous property of $f(\thetabf;\cdot)$, we must also have $\|\thetabft\|_2 \to \infty$, which completes \textbf{S1}.

Now we show the results for \textbf{S2}. We design the scaling factor to be: $\alpha(t) = \min_{u,i}q_{u,i}(\thetabft)^{1/L}$. Consequently, $\nabla_{\thetabf}q_{u,i}(\thetabft) / \alpha(t)^{L-1} = \nabla_{\thetabf}q_{u,i}\big(\alpha(t)\thetabft \big)$. 

We then construct the Lagrange multipliers as:
\[
\lambda_{u,i}(t) = \alpha(t)^{L-2}\|\thetabft\|_2 \exp\big( -q_{u,i}(\thetabft)\big) \big/ \|\frac{d\thetabft}{dt}\|_2.
\]
Using the results in \text{S1}, and by straightforward calculations, we obtain:
\begin{equation}
\label{eqn:KKT1}
\Big\|\alpha(t) \thetabft - \sum_{(u,i)\in\Dcal_{\text{train}}}\lambda_{u,i}y_{u,i}\nabla f\big(\alpha(t)\thetabft;(u,i)\big) \Big\|_2^2 \leq \frac{2}{\tgamma\big(\thetabft\big)^{2/L}}\bigg( 1- \Big\langle\frac{\thetabft}{\|\thetabft\|_2}, \frac{d\thetabft/dt}{\|d\thetabft/dt\|_2} \Big\rangle \bigg)
\end{equation}
and 
\begin{equation}
\label{eqn:KKT2}
\lambda_{u,i}(t) \Big(y_{u,i}f\big(\alpha(t)\thetabft;(u,i)\big) - 1 \Big) \lesssim \frac{2e|\Dcal_{\text{train}}|}{L\tgamma\big(\thetabft\big)^{2/L+1} \|\thetabft\|_2^2}.
\end{equation}

According to the game plan, we then need to show $\Big\langle\frac{\thetabft}{\|\thetabft\|_2}, \frac{d\thetabft/dt}{\|d\thetabft/dt\|_2} \Big\rangle \to 1$ to show (\ref{eqn:KKT1})$\to 0$, since (\ref{eqn:KKT2})$\to 0$ is implied by \text{S1}. First notice that $\langle\frac{\thetabft}{\|\thetabft\|_2}, \frac{d\thetabft/dt}{\|d\thetabft/dt\|_2} \Big\rangle \leq 1$ by the Cauchy-Schwartz inequality. We then need to show it is $\geq 1$ as $t\to\infty$.
Using the results in \textbf{S1}, we have:
\begin{equation}
\label{eqn:S3A}
\begin{split}
\frac{d\tgamma\big(\thetabft\big)}{dt} & \geq \frac{\|\thetabft\|_2^2}{Q(t)} \cdot \Big\|\frac{d\tilde{\thetabft}}{dt} \Big\|_2^2 = L\Big\|\frac{\|\thetabft\|_2^2}{LQ(t)} \frac{d\tilde{\thetabft}}{dt}\Big\|_2^2\cdot \frac{LQ(t)}{\|\thetabft\|_2^2} \\
&=L\Big\|\frac{\|\thetabft\|_2^2}{LQ(t)}\cdot \frac{d\tilde{\thetabft}}{dt}\Big\|_2^2\cdot \frac{d\log\|\thetabft\|_2}{dt} .
\end{split}
\end{equation}
Since 
\begin{equation}
\label{eqn:S3B}
\begin{split}
    \Big\|\frac{\|\thetabft\|_2^2}{LQ(t)} \frac{d\tilde{\thetabft}}{dt}\Big\|_2^2 & = \Big\|\frac{\|\thetabft\|_2^2}{LQ(t)} \cdot \frac{1}{\|\thetabft\|_2}\Big(\mathbf{I} - \tilde{\thetabft}\tilde{\thetabft}^{\intercal} \Big)\frac{d\thetabft}{dt} \Big\|_2^2 \\
    & = \frac{\Big\|\frac{d\thetabft}{dt}\Big\|_2^2 - \Big\langle\tilde{\thetabft}, \frac{d\thetabft}{dt} \Big\rangle^2}{\Big\langle\tilde{\thetabft}, \frac{d\thetabft}{dt} \Big\rangle^2} \\
    & = \bigg\langle \frac{\thetabft}{\|\thetabft\|_2}, \frac{d\thetabft/dt}{\|d\thetabft/dt\|_2} \bigg\rangle^{-2}-1,
\end{split}
\end{equation}
combining (\ref{eqn:S3A}) and (\ref{eqn:S3B}), we have:
\[
\bigg\langle \frac{\thetabft}{\|\thetabft\|_2}, \frac{d\thetabft/dt}{\|d\thetabft/dt\|_2} \bigg\rangle \geq \sqrt{1+\frac{d\log \tgamma\big(\thetabft\big)/dt}{L\cdot d\log\|\thetabft\|_2/dt}} \geq \sqrt{1+\frac{\epsilon(t)}{L}}\, \text{ for some } \epsilon(t)\geq 0,
\]
because both $d\log\|\thetabft\|_2/dt \geq 0$ and $d\log \tgamma\big(\thetabft\big)/dt \geq 0$. Hence, by the previous argument, we have $\lim_{t\to \infty} \Big\langle \frac{\thetabft}{\|\thetabft\|_2}, \frac{d\thetabft/dt}{\|d\thetabft/dt\|_2} \Big\rangle = 1$. By showing the resutls in \textbf{S1}, \textbf{S2} and \textbf{S3}, we see that $\alpha \thetabft$ converges in direction to the KKT points, which completes the proof.

\end{proof}

\textbf{Proof for Theorem 2.}
\begin{proof}
The first part of the results for NCF is a direct consequence of Corollary \ref{corollary:max-margin}. To show the second part for MCF, we need to consider the symmetrized setting with:
\[
\Wbf:= \Zbf_U\Zbf_I^{\intercal}, \quad \tilde{\Wbf} = 
\begin{bmatrix}
\Mbf_1 & \Wbf \\
\Wbf & \Mbf_2
\end{bmatrix}, \quad \tilde{\Xbf}_{u,i} = \ebf_u \ebf_i^{\intercal} +  \ebf_i \ebf_u^{\intercal},
\]
where $\Mbf_1$ and $\Mbf_2$ are two p.s.d matrices that are irrelevant for the objective, and the definition of $ \ebf_u,\ebf_i$ are provided in \ref{sec:sec4}. Notice that in the main paper, we use $\Xbf$ to denote the predictor which is now given by $\Wbf$. Hence, the MF parameterization can be considered by: $\tilde{\Zbf}\tilde{\Zbf}^{\intercal} = \tilde{\Wbf}$, and the objective becomes:
\[
\min_{\tilde{\Zbf}} \Lcal(\tilde{\Zbf}) = \sum_{(u,i)\in\Dcal_{\text{train}}} \ell\Big(-y_{u,i} \Big\langle \tilde{\Zbf} \tilde{\Zbf}^{\intercal}, \tilde{\Xbf}_{u,i} \Big\rangle \Big).
\]
It is easy to verify that the symmetrized MCF corresponds exactly to the original problem instance, and satisfy the conditions in Corollary \ref{corollary:max-margin} under exponential or log loss. Since $\lim_{t\to\infty} \frac{\Wbf^{(t)}}{\big\|\Wbf^{(t)} \big\|_{*}} = \lim_{t\to\infty} \frac{\tilde{\Zbf}^{(t)}}{\big\|\tilde{\Zbf}^{(t)} \big\|_F} \frac{\tilde{\Zbf}^{(t)}}{\big\|\tilde{\Zbf}^{(t)} \big\|_F}^{\intercal}$, and the MCF predictor is convex, we conclude that the predictor of MCF converges in direction to the stationary point of:
\begin{equation}
\label{eqn:nuc-svm}
    \min \, \|\Wbf\|_* \quad s.t. \quad y_{u,i}\Wbf_{u,i} \geq 1, \, \forall (u,i)\in \Dcal_{\text{train}}.
\end{equation}

\end{proof}

\section{Material for Section 6}
\label{sec:sec6}

The major tools we use to show the generalization results are the Rademacher complexities. The procedure of bounding the inductive generalization error via the symmetrization technique and Talagrand's contraction inequalities are more often encountered in the literature \cite{bartlett2002rademacher}. The similar ideas can also be applied to bound the transductive generalization error, but specific modifications are required \cite{el2009transductive}. 

The different meaning of generalization decides the distinctive definitions of Rademacher complexities. We use $\mathcal{X}$ to denote the domain (of user and items) for the CF predictors, $n_1$ to denote $|\Dcal_{\text{train}}|$ and $n_2$ to denote $|\Dcal_{\text{test}}|$. We first provide the definitions of  Rademacher complexities, and briefly discuss their different implications for the transductive and inductive learning.

\begin{definition}
\label{def:rademacher}
Recall that $n_1=|\Dcal_{\text{train}}|$ and $n_2=|\Dcal_{\text{test}}|$.

\begin{itemize}
    \item \textbf{Transductive Rademacher complexity}. Let $\mathcal{V} \in \Rbb^{n_1+n_2}$ and $p\in[0,1/2]$, and $\epsilon_i(p)$ be i.i.d random variables such that:
    \[
    \epsilon_i(p) = \left\{ \begin{array}{lll}
        1 & \text{with probability} & p \\
        -1 & \text{with probability} & p \\
        0 & \text{with probability} & 1-2p, \\
    \end{array}
    \right.
    \]
    then the trasductive Rademacher complexity of $\mathcal{V}$ is:
    \begin{equation}
        \label{eqn:trans-rc}
        R_{n_1+n_2}(\mathcal{V},p) = \Big( \frac{1}{n_1} + \frac{1}{n_2} \Big)\Ebb\Big\{\sup_{\vbf \in \mathcal{V}} \epsilonbf(p)^{\intercal}\vbf \Big\},
    \end{equation}
    where $\epsilonbf(p)=\big[\epsilon_1(p), \ldots, \epsilon_{n_1+n_2}(p)\big]^{\intercal}$. 
    
    \item \textbf{Inductive Rademacher complexity}. Let $\mathcal{F}$ be a function class with domain $\mathcal{X}$, and $\{X_{i}\}$ be a set of samples generated by a distribution $P_{\mathcal{X}}$ on $\mathcal{X}$. Let $\epsilon_i$ be shorthand of the same i.i.d random variables as above, with $p=1/2$. Then the empirical Rademacher complexity of $\mathcal{F}$ is:
    \[
    \hat{R}_n(\mathcal{F}) = \Ebb\Big\{\sup_{f\in\mathcal{F}} \Big|\frac{1}{n}\sum_{i=1}^n\epsilon_i f(X_i) \Big| | X_1,\ldots,X_n  \Big\},
    \]
    and the Rademacher complexity is given by: $R_n(\Fcal) = \Ebb_{P_{\mathcal{X}}}\hat{R}_n(\mathcal{F})$.
\end{itemize}
\end{definition}

A important difference between the two settings is that the transductive complexity is an empirical quantity that does not depend on any underlying distributions, and it depends on both the training and testing data. The other difference is reflected in the specific formulations are:

1. transductive Rademacher complexity depends on both $n_1$ and $n_2$ because of the need to bound the test error: $\Dcal_{\text{test}}$, i.e.  $\text{Err}_{\Dcal_{\text{test}}}(f):=\sum_{(u,i) \in \Dcal_{\text{test}}} 1\big[ y_{u,i} f(u,i) \leq 0 \big]$;

2. it depends only on the outcomes' vector space rather than the underlying function space that produces the outcomes.

The different definitions of Rademacher complexity induces the two versions of contraction inequalities, which we provide in Lemma \ref{lemma:inductive-contraction} and Lemma \ref{lemma:transductive-contraction}. We first prove the results in Theorem 3 for the transductive setting. Using the idea of symmetrization, the bound on the testing error for the transductive learning can be stated as in the following Corollary. 

\begin{corollary}[Adapted from \citet{el2009transductive}]
\label{corollary:transductive}
Let $\Vcal$ be a set of real-valued vectors in $[-B, B]^{n_1+n_2}$, where $n_1>n_2$ by our assumption. Define $Q=(1/n_1 + 1/n_2)$, $S=\frac{n_1+n_2}{(n_1+n_2-1/2)(1-n_1/2)}$. Then for all $\vbf\in\Vcal$, with probability of at least $1-\delta$ over the random permutation of $\vbf$, which we deonote by $\tilde{\vbf}$, we have:
\[
\sum_{j=n_1+1}^{n_1+n_2} \tilde{\vbf}_j \leq \sum_{j=1}^{n_1} \tilde{\vbf}_j + R_{n_1+n_2}(\Vcal, p_0) + Bc_0Q\sqrt{n_2} + B\sqrt{\frac{S}{2}Q\log\frac{1}{\delta}},
\]
where $c_0 = \sqrt{32\log(4e)/3}$ and $p_0=n_1n_2/(n_1+n_2)^2$.
\end{corollary}

By defining the $\Vcal$ in the above corollary by the scores of the predictor, and using Lemma \ref{lemma:transductive-contraction} for contraction, we are able to show the results in Theorem 3.

\textbf{Proof for Theorem 3.}
\begin{proof}
Define $\hbf \in \Hcal_{out} \in \Rbb^{n_1+n_2}$ as the output scores of the predictor, and consider $\vbf$ in Corollary \ref{corollary:transductive} as $\ell\big(y_{u,i} f(\thetabf; (u,i))\big)$ where $\ell(u) = 1[u < 0]$. Define $\ell_{\gamma}(y_{u,i} f(\thetabf, \xbf_{u,i}))$ to be the margin loss: $\ell_{\gamma}(u) = \min\{1, 1-u/\gamma\}$. Note that the margin loss is an upper bound on the classification error.

Therefore, using the results in Corollary \ref{corollary:transductive} and Lemma \ref{lemma:transductive-contraction}, for any fixed $\gamma>0$ and $\hbf \in \Hcal_{out}$, with probability of at least $1-\delta$ over the random splits of $\Dcal$:
\begin{equation*}
\begin{split}
    \frac{1}{n_2}\sum_{(u,i)\in\Dcal_{\text{test}}}1\big[y_{u,i} f\big(\thetabf, (u,i)\big) < 0 \big] & \leq \frac{1}{n_1}\sum_{(u,i)\in\Dcal_{\text{train}}}1\big[y_{u,i} f\big(\thetabf, (u,i)\big) < \gamma \big] \\
    & + \frac{R_{n_1+n_2}(\Hcal_{out},p_0)}{\gamma} +c_0Q\sqrt{n_2} + \sqrt{\frac{S}{2}Q\log\frac{1}{\delta}} .
\end{split}
\end{equation*}

We first show the bound for the transductive Rademacher complexity for NCF. 
Recall that $\Hcal_{out}$ for NCF is given by the form of: $\Wbf_1 \sigma(\Wbf_2 \sigma(\ldots \sigma(\Wbf_q \sigma(\Wbf_{q+1} \Xbf_{u,i}))))$, where $\Wbf_{q+1}$ is given by (\ref{eqn:ncf-a}) or (\ref{eqn:ncf-c}), with $\max_{(u,i)\in \Dcal}\|\zbf_u+\zbf_i\|_2 \leq B_{\text{NCF}}$ for NCF with addition, and $\|\Wbf_i\|_F \leq \lambda_i$ for $i=1,\ldots,q$. We denote the output of the $k^{th}$ layer by $\Hbf_{out}^k \in \Hcal_{out}^k$. It holds that:
\begin{equation}
\label{eqn:transductive-ncf}
\begin{split}
& R_{n_1+n_2}(\Hcal_{out}, p_0)  = \big(\frac{1}{n_1} + \frac{1}{n_2}\big) \Ebb\Big\{\sup_{\|\Wbf_1\|_F \leq \lambda_1} \sum_{(u,i)\in\Dcal_{\text{train}}} \epsilon_{u,i} \big[\Wbf_1 \Hbf_{out}^{(q-1)}\big]_{u,i} \Big\} \\
& \leq \lambda_1 \big(\frac{1}{n_1} + \frac{1}{n_2}\big) \Ebb\Big\{\sup_{\|\Wbf_2\|_F \leq \lambda_2} \sum_{(u,i)\in\Dcal_{\text{train}}} \epsilon_{u,i} \big[\Wbf_2 \Hbf_{out}^{(q-2)}\big]_{u,i} \Big\}\, \text{(applying Lemma \ref{lemma:transductive-contraction} on ReLU)} \\
& \text{recursively apply the peeling argument} \\
& \leq \prod_{i=1}^q \lambda_i \big(\frac{1}{n_1} + \frac{1}{n_2}\big) \Ebb\Big\{ \sup_{\max \|\zbf_u+\zbf_i\|_2 \leq B_{\text{NCF}}} \big\langle \zbf_u+\zbf_i, \sum_{(u,i)\in\Dcal_{\text{train}}} \epsilon_{u,i}\Xbf_{u,i} \big\rangle \Big\} \text{( under addition, for example)} \\
& \leq B_{\text{NCF}}\prod_{i=1}^q \lambda_i \big(\frac{1}{n_1} + \frac{1}{n_2}\big) \Ebb \Big\|\sum_{(u,i)\in\Dcal_{\text{train}}} \epsilon_{u,i}\Xbf_{u,i} \Big\|_2,
\end{split}
\end{equation}
where we use $\epsilon_{u,i}$ as a shorthand for $\epsilon_{u,i}(p_0)$.
By Jensen's inequality, the last line is upper-bounded by:
\[
B_{\text{NCF}}\prod_{i=1}^q \lambda_i \big(\frac{1}{n_1} + \frac{1}{n_2}\big) \sqrt{\sum_{(u,i)\in\Dcal_{\text{train}}} \Ebb\big[\epsilon_{u,i}(p_0)^2 \big] \|\Xbf_{u,i}\|_2} \leq B_{\text{NCF}}\prod_{i=1}^q \lambda_i\frac{n_1+n_2}{n_1n_2},
\]
where it is easy to compute that: $\Ebb\big[\epsilon_{u,i}(p_0)^2 \big] = \frac{2n_1n_2}{(n_1+n_2)^2}$.
By plugging in the relation between $n_1$ and $n_2$, we obtain the result stated in Theorem 3 for NCF. 

We then show the results for MCF. We use $\Sigmabf(p)$ to denote the matrix of transductive Rademacher random variables such that $\Sigmabf_{u,i}(p) = \epsilon_{u,i}(p)$. For $\Hbf:= \Zbf_U\Zbf_I^{\intercal} \in \Hcal_{out}$ under $\big\|\Zbf_U\Zbf_I^{\intercal}\big\|_{*}\leq \lambda_{\text{nuc}}$, we have:
\begin{equation}
\label{eqn:transductive-mcf}
\begin{split}
    R_{n_1+n_2}(\Hcal_{out},p_0) & = \big(\frac{1}{n_1} + \frac{1}{n_2}\big) \Ebb\Big\{ \sup_{\Hbf: \|\Hbf\|_* \leq \lambda_{\text{nuc}}} \sum_{(u,i)\in\Dcal_{\text{train}}} \Sigmabf_{u,i}\Hbf_{u,i} \Big\} \\
    & \leq \lambda_{\text{nuc}} \big(\frac{1}{n_1} + \frac{1}{n_2}\big) \Ebb \big\| \Sigmabf \big\|_{sp}\, \text{(by H\"{o}lder inequality, where $\|\cdot\|_{sp}$ is the spectral norm)} \\
    & \lesssim \lambda_{\text{nuc}} \frac{(n_1+n_2)\sqrt{|\Ical|} \sqrt[4]{\log |\Ucal|}}{n_1n_2}\, \text{(by Lemma \ref{lemma:mat-ineq}).}
\end{split}
\end{equation}
The second line holds because nuclear norm is the dual of the spectral norm.
Again we plug in the relation between $n_1$ and $n_2$ and obtain the stated result for MCF. 
\end{proof}

We move on to proving the generalization results for the inductive CF. We first state a useful corollary for inductive learning, when the training and testing distribution are different. 

\begin{corollary}
\label{corollary:weighted_kakade2009}
  Consider an arbitrary function class $\Fcal$ such that $\forall f \in \Fcal$ we have $\sum_{\xbf \in \Xcal}|f(\xbf)| \leq C$. Then, with probability at least $1 - \delta$ over the sample, for all margins $\gamma > 0$ and all $f \in \Fcal$ we have,
  \begin{equation}
  \begin{split}
      & P_{\text{test}}\Big(yf(\xbf) \leq  0\Big) \\ 
      & \leq \frac{1}{n}\sum_{i=1}^n \eta(\xbf_{u,i}) 1\big(y_i f(\xbf_{u,i}) < \gamma \big) + 4\frac{R_{n, \etabf}(\Fcal)}{\gamma} + \sqrt{\frac{\log(\log_2\frac{4C}{\gamma})}{n}} + \sqrt{\frac{\log(1/\delta)}{2n}},
\end{split}
\end{equation}
  where $\eta(\xbf_{u,i}) = P_{\text{test}}(\xbf_{u,i}) / P_{\text{train}}(\xbf_{u,i})$ gives the importance weighting, and $R_{n, \etabf}(\Fcal) = \Ebb \Big[ \sup_{f\in \Fcal}\frac{1}{n}\sum_{i = 1}^n \eta(\xbf_{u,i}) f(\xbf_{u,i})\epsilon_i\Big]$ is the weighted Rademacher complexity.
\end{corollary}
\begin{proof}
This corollary is adapted from the more general Theorem 1 of \cite{koltchinskii2002empirical} by considering the deviation of the testing distribution from the training distribution. The stated result is then obtained following the Theorem 5 of \cite{kakade2008complexity}.
\end{proof}

Therefore, the key step for proving the results in Theorem 4 is to bound the weighted Rademacher complexity for NCF and MCF. 

\textbf{Proof for Theorem 4.}
\begin{proof}
We first show the results for NCF, where we denote the predictor family by $\Fcal_{\text{NCF}}$. Here, using the similar setup from Theorem 1 of \cite{golowich2018size}, and combining the same arguments from the proof of Theorem 3, we arrive at
  $$n_1R_{n_1, \etabf}(\Fcal_{\text{NCF}}) \leq \frac{1}{\lambda} \log \Big (2^q \cdot \Ebb_{\epsilonbf} \Big (M\lambda \Vert\sum_{i = 1}^{n_1} \epsilon_i \eta(\xbf_{u,i}) \xbf_{u,i}\Vert \Big) \Big ),$$
  where $M = B_{\text{NCF}}\prod_{h = 1}^q \lambda_i$. Consider $Z := M \cdot \Vert\sum_{i = 1}^{n_1} \epsilon_i \eta(\xbf_{u,i}) \xbf_{u,i}\Vert$ that is a random function of the $n_1$ Rademacher variables. Then
  $$\frac{1}{\lambda} \log \Big \{2^q \Ebb \exp (\lambda Z) \Big \} = \frac{q \log (2)}{\lambda} + \frac{1}{\lambda} \log \{\Ebb \exp \lambda (Z - \Ebb Z)\} + \Ebb Z.$$
  By Jensen's inequality, we have
  $$\Ebb [Z] \leq M \sqrt{\Ebb_{\epsilonbf} \Vert\sum_{i = 1}^{n_1} \epsilon_i \eta(\xbf_{u,i}) \xbf_{u,i} \Vert^2} = M \sqrt{\sum_{i = 1}^{n_1} \eta(\xbf_{u,i})^2 \Vert \xbf_{u,i} \Vert^2}.$$
  In addition, we note that
  $$Z(\epsilon_1, \ldots, \epsilon_i, \ldots, \epsilon_{n_1}) - Z(\epsilon_1, \ldots, -\epsilon_i, \ldots, \epsilon_{n_1}) \leq 2M\eta(\xbf_{u,i}) \Vert \xbf_{u,i} \Vert.$$
  By the bounded-difference condition \citep{boucheron2013concentration}, $Z$ is a sub-Gaussian with variance factor $v = \frac{1}{4} \sum_{i = 1}^{n_1} (2M\eta(\xbf_{u,i}) \Vert \xbf_{u,i}\Vert)^2 = M^2 \sum_{i = 1}^{n_1} \eta(\xbf_{u,i})^2 \Vert \xbf_{u,i} \Vert^2$. So
  $$\frac{1}{\lambda} \{\Ebb \exp \lambda (Z - \Ebb Z)\} \leq \frac{\lambda M^2 \sum_{i = 1}^{n_1} \eta(\xbf_{u,i})^2\Vert \xbf_{u,i} \Vert^2}{2}.$$
  Taking $\lambda = \frac{\sqrt{2\log (2) q}}{M \sqrt{\sum_{i = 1}^{n_1} \eta(\xbf_{u,i})^2 \Vert \xbf_{u,i} \Vert^2}}$, it follows that
  \begin{equation}
  \begin{split}
  & \frac{1}{\lambda} \{2^q \Ebb \exp \lambda Z\} \\ 
  & \leq M (\sqrt{2 \log (2) q} + 1) \sqrt{\sum_{i = 1}^{n_1} \eta(\xbf_{u,i})^2 \Vert \xbf_{u,i} \Vert^2} \leq \sqrt{n_1} CM(\sqrt{2 \log (2)q} + 1) \sqrt{\frac{1}{n_1}\sum_{i = 1}^{n_1} \eta(\xbf_{u,i})^2},
  \end{split}
  \end{equation}
  where $C=1$ for NCF-c and $C=\sqrt{2}$ for NCF-a. 
By law of large number, $\frac{1}{n_1} \sum_{i = 1}^{n_1} \eta(\xbf_{u,i})^2 = D(P_{\text{test}}||P_{\text{train}}) + 1 + o(\frac{1}{\sqrt{n_1}})$. The desired result for NCF follows. 

Then we show the result for MCF. 
Here, we provide a general result for generalization of using importance weighting under distribution shift. We assume the training distribution is $P$, testing distribution is $Q$, and the weight for any $(u,i)$ instance is therefore given by: $w_i = Q(i)/P(i)$. We define $\Ncal(\frac{1}{n},\Fcal,\ell_{2}^{n})$ as the $\frac{1}{n}$-covering number for $\Fcal$ in $\|\cdot\|_{2}$ based on $n$ i.i.d samples from $P$, and $d(P\|Q) = \int_{\Scal_Q} (dP / dQ) dP$ is a divergence measure, where $\Scal$ is used to denote the support of a distribution. We use $\Ebb_{Q}R(f)$ to denote the testing risk, and use $\Ebb_{P_{n,w}}R(f)$ to denote the weighted empirical training risk.

Our proof leverages the classical "double sampling" technique from \citet{anthony2009neural}. We use $\vec{\zbf} = [\zbf_1,\ldots,\zbf_n]$ to denote the observed samples, and $\vec{\zbf}' = [\zbf'_1,\ldots,\zbf'_n]$ to denote an i.i.d copy of $\vec{\zbf}$. We first define by: 
\[
UB_1(f, \vec{\zbf}, t) = \frac{1}{n}\sum_{i=1}^n w_i \ell_f(\zbf_i) + \frac{3Mt}{n} + \sqrt{\frac{2d(P\|Q)t}{n}},
\]
and 
\[
UB_2(f, \vec{\zbf}, t) = \frac{1}{n}\sum_{i=1}^n w_i \ell_f(\zbf_i) + \frac{9Mt}{n} + \sqrt{\frac{18d(P\|Q)t}{n}}.
\]
Given $f\in\Fcal$, let $A := \Ebb_{Q}R(f) + \frac{6Mt}{n} + \sqrt{\frac{8d(P\|Q)t}{n}} $it holds:
\begin{equation*}
\begin{split}
    \Pbb\big( UB_2(f, \vec{\zbf}', t) \leq UB_1(f, \vec{\zbf}, t) \big) & \leq \Pbb\big(UB_2(f, \vec{\zbf}', t) \leq A \big) + \Pbb\big(UB_1(f, \vec{\zbf}, t) \geq A \big) \\
    &\leq 2\Pbb\Big( \big| \Ebb_{Q}R(f) - \frac{1}{n}\sum w_i\ell_f(\zbf_i) \big| \geq \frac{3Mt}{n} + \sqrt{\frac{2d(P\| Q)t}{n}} \Big) \\
    &\leq 4e^{-t},
\end{split}
\end{equation*}
where the last line follows from Lemma \ref{lemma:append-lemma-importance-weighting}. 
Next, we define $\Ccal(\epsilon, \ell\circ\Fcal, \ell_1(P_{n,w}))$ be the $\epsilon$-cover of $\ell\circ\Fcal$ with the empirical $\ell_1$ norm under $P_{n,w}$ such that for any $f\in\ell\circ\Fcal$, there exists $\tilde{f}$ in $\Ccal(\epsilon, \ell\circ\Fcal, \ell^n_1)$: $\big|\frac{1}{n}\sum w_i f(\zbf_i) -  \frac{1}{n}\sum w_i \tilde{f}(\zbf_i) \big| \leq \epsilon$, for $(\zbf_1,\ldots,\zbf_n)$ sampled i.i.d from $P$.
It then holds:
\begin{equation*}
\begin{split}
    &\Pbb\big(\exists f\in\Fcal: \Ebb_{Q}R(f) \geq UB_2(f, \vec{\zbf}, t) + \epsilon \big) \\
    &= \Ebb_{\vec{\zbf}} \sup_{f\in\Fcal} I[\Ebb_Q R(f) \geq UB_2(f, \vec{\zbf}, t) + \epsilon ] \\
    &\overset{(a)}{\leq} \Ebb_{\vec{\zbf}} \sup_{f\in\Fcal} I[\Ebb_Q R(f) \geq UB_2(f, \vec{\zbf}, t) + \epsilon ]\cdot 2\Ebb_{\vec{\zbf}'}I[UB_1(f, \vec{\zbf}', t) \geq \Ebb_{Q}R(f)] \\
    &\leq 2 \Ebb_{\vec{\zbf}, \vec{\zbf}'} \sup_{f\in\Fcal}I[UB_1(f, \vec{\zbf}', t) \geq UB_2(f, \vec{\zbf}, t) + \epsilon] \\
    &\overset{(b)}{\leq} 2 \Pbb_{\sigma(\vec{\zbf}, \vec{\zbf}')}\big(\exists \tilde{f} \in \Ccal(\epsilon, \ell\circ\Fcal, \ell_1(P_{n,w})):\, UB_1(f, \sigma(\vec{\zbf}, \vec{\zbf}'), t) \geq UB_2(f, \sigma(\vec{\zbf}, \vec{\zbf}'), t) \big) \\
    &\leq 8\Ncal(\epsilon, \ell\circ\Fcal,\ell_1(P_{n,w})) \cdot e^{-t},
\end{split}
\end{equation*}
where (a) follows from the fact that $\Ebb_{\vec{\zbf}'}I[UB_1(f, \vec{\zbf}', t) \geq \Ebb_{Q}R(f)]\geq \frac{1}{2}$ as suggested by Lemma \ref{lemma:append-lemma-importance-weighting}, and in step (b) we let $\sigma(\vec{\zbf}, \vec{\zbf}')_i$ takes the value of $\zbf_i,\zbf'_i$ with equal probability, and the inequality follows from the definition of the $\epsilon$ cover. Notice that $\Ncal(\epsilon, \ell\circ\Fcal,\ell_1(P_{n,w})) \leq \Ncal(\epsilon/M, \ell\circ\Fcal, \ell^{n}_{2})$. We take $\epsilon = \frac{1}{n}$, which solves for $t = c\log\frac{1}{\delta} + \log\Ncal(\epsilon/M, \ell\circ\Fcal, \ell^{n}_{2})$ for some constant $c$. We use $\Ncal_{2}(\epsilon, \Fcal\big)$ as a shorthand to denote the covering number under the empirical $\ell_2$ norm.

By rearranging terms, we have that for any $\delta>0$, with probability at least $1-\delta$, it holds:
    \begin{equation}
    \label{eqn:append-covering-weighting}
    \begin{split}
        \Ebb_{Q}R(f) \lesssim \Ebb_{P_{n,w}}R(f) + \frac{M\big(\log\frac{1}{\delta} + \log\Ncal_{2}(\frac{1}{n}, \Fcal\big)}{n} + \sqrt{\frac{Md(P\|Q)\big(\log\frac{1}{\delta} + \log\Ncal_{2}(\frac{1}{n}, \Fcal\big)}{n}}, \\
    \end{split}
    \end{equation}
when the loss function $\ell$ is Lipschitz and we ignore the constants. Hence, the remaining task is to bound the covering number of the matrix factorization class $\Fcal_{\text{MCF}}$ with a bounded nuclear norm. When $\|\Xbf\|_F = 1$, the nunclear norm is strictly a lower bound of the matrix rank. Therefore, we use the covering number of low-rank matrix as a upper bound, which according to Lemma 3.1 of , if $\text{rank}(\Xbf)\leq\lambda_{\text{nuc}}$, then the covering number for $\Fcal_{\text{MCF}}$ under the matrix Frobenius norm obeys:
\[
\Ncal(\epsilon, \Fcal_{\text{MCF}}, \|\cdot\|_F) \leq (9/\epsilon)^{(|\Ucal| + |\Ical| + 1) \lambda_{\text{nuc}}},
\]
which we plug back to (\ref{eqn:append-covering-weighting}) and obtain the desired result.
\end{proof}

\textbf{Discussion: the tightness of the generalization bounds.}

When proving the bounds for both the transductive and inductive CF, we use the standard generalization results based on Rademacher complexity, according to \citet{bartlett2002rademacher} and \citet{el2009transductive}. Their results rely on the following components: 
\begin{itemize}
    \item a symmetrization argument to bound the testing error;
    \item the Mcdiarmid's inequality for bounded difference;
    \item the Rademacher contraction inequalities (Lemma \ref{lemma:inductive-contraction} and Lemma \ref{lemma:transductive-contraction}).
\end{itemize}
All these results are known to be tight, so the question narrows down to the tightness of our bounds on the Rademacher complexities. To see that the provided result for NCF are tight up to a constant factor of $\sqrt{q}$, we simply consider the following construction: $\xbf_{u,i} \mapsto \lambda_1\cdot\lambda_2 \cdots \lambda_q \cdot \sigma(\Wbf_{q+1}\xbf_{u,i})$, which belongs to the general NCF family, and the worst-case scenario for computing Rademacher complexity is obvious given by: 
\[\lambda_1\cdot\lambda_2 \cdots \lambda_q \cdot \sigma\Big(\max_{u,i: \|\zbf_u + \zbf_i\|_2} \zbf_u + \zbf_i\Big),
\]
where we use NCF-a for example. Here, all the training samples are $(u,i) = \arg\max_{u,i: \|\zbf_u + \zbf_i\|_2}$. Consequently, the Rademacher complexity is at least $B_{\text{NCF}} \prod_{i=1}^q \lambda_i$. 

On the MCF side, it is pointed out by \cite{bandeira2016sharp} that for the spectral norm of Rademacher matrix, the dependency on $\sqrt{|\Ical|}$ and $\sqrt[4]{\log|\Ucal|}$ are inevitable, and therefore our result for transdutive MCF is also tight up to constants. As for the inductive setting, we refer to the results in \citet{candes2009exact} that the bound with $\sqrt{(\sqrt{|\Dcal|}+ \sqrt{|\Ucal|})}\big/\sqrt{n_1}$ is not improvable. Notice that they assume a uniform distribution over the matrix indices, where our result is distribution-free. However, despite several minor discrepancies, their setting can be recognized as a special case of our problem, and thus we conjecture that our results for MCF can be further tightened to get rid of the $\log n$ dependency, e.g. by deriving the covering number for nuclear-norm-constraint matrices instead of using the existing result for low-rank matrices.

\section{Auxiliary Lemmas}
\label{sec:lemma}

\begin{lemma}[Adapted from \citet{cho2009kernel}]
\label{lemma:arc-cos-kernel}
Define the shorthand $\varsigma(u) := \frac{1}{2}(1+sign(u))$. For $\xbf, \ybf \in \Rbb^d$, the $n^{th}$ order arc-cosine kernel is defined as:
\[
K_n(\xbf, \ybf) = \frac{1}{\pi}\|\xbf\|_2^n \|\ybf\|_2^n J_n(\theta),
\]
where $J_n(\theta) = (-1)^n(\sin \theta)^{2n+1}\big(\frac{1}{\sin\theta}\frac{d}{d\theta} \big)^n \big(\frac{\pi-\theta}{\sin\theta} \big)$. Then the arc-cosine kernel has an equivalent integral representation:
\[
K_n(\xbf, \ybf) = 2\int d\wbf \frac{\exp(-\frac{\|\wbf\|_2^2}{2})}{(2\pi)^{d/2}}\varsigma(\wbf^{\intercal}\xbf) \varsigma(\wbf^{\intercal}\ybf) (\wbf^{\intercal}\xbf)^n (\wbf^{\intercal}\ybf)^n.
\]
\end{lemma}

For instance, when $n=0$, $K_0(\xbf, \ybf) = 1- \frac{1}{\pi}\cos^{-1}\frac{\xbf^{\intercal}\ybf}{\|\xbf\|_2 \|\ybf\|_2}$.

\begin{lemma}[Euler's Theorem for homogeneous functions]
\label{lemma:homogeneous}
If $f(\thetabf,\cdot)$ is L-homogeneous, then:
\begin{itemize}
    \item $\nabla f(\alpha \thetabf,\cdot) = \alpha^{L-1}\nabla f(\thetabf,\cdot)$,
    \item $\big\langle \thetabf,  \nabla f(\thetabf,\cdot) \big\rangle = L\cdot f(\thetabf,\cdot),$
\end{itemize}
if $f(\thetabf,\cdot)$ is differentiable. 
\end{lemma}
The proof for the Lemma is relatively standard, so we do not repeat it here.

\begin{lemma}[\citet{ledoux1991probability}]
\label{lemma:inductive-contraction}
Let $f:\Rbb_+ \to :\Rbb_+$ be convex and increasing. Let $\phi_i: \Rbb \to \Rbb$ satisfy $\phi_i(0)=0$ and is Lipschitz with constant L. Then for any $\Vcal \in \Rbb^n$:
\[
\Ebb f\Big(\frac{1}{2}\sup_{\vbf \in \Vcal} \Big|\sum_{i=1}^n\epsilon_i\phi_i(\vbf_i) \Big| \Big) \leq \Ebb f\Big(L\cdot\frac{1}{2}\sup_{\vbf \in \Vcal} \Big|\sum_{i=1}^n\epsilon_i\vbf_i \Big|\Big),
\]
where $\epsilon_i$ are the standard Rademacher random variables.
\end{lemma}
The similar contraction result in the transductive setting is given as below.

\begin{lemma}[Lemma 5 of \citet{el2009transductive}]
\label{lemma:transductive-contraction}
Consider $\Vcal \in \Rbb^{n_1+n_2}$. Let $f, g: \Rbb \to \Rbb$ be such that for all $1\leq i \leq n_1+n_2$ and $\vbf,\vbf' \in \Vcal$, $\big|f(\vbf_1) - f(\vbf'_1)\big| \leq L\big|g(\vbf_1) - g(\vbf'_1)\big|$, then:
\[
\Ebb\Big\{\sup_{\vbf \in \Vcal}\sum_{i=1}^{n_1+n_2}\epsilon_i(p)f(\vbf_i)\Big\} \leq \Ebb\Big\{L\cdot \sup_{\vbf \in \Vcal}\sum_{i=1}^{n_1+n_2}\epsilon_i(p)g(\vbf_i)\Big\},
\]
for any $p\in[0,1/2]$.
\end{lemma}

\begin{lemma}[Concentration of random matrices.]
\label{lemma:mat-ineq}
Let $\Xbf$ be a $m\times n$ matrix with $m>n$. 
\begin{itemize}
    \item By \citet{bandeira2016sharp}, if $\Xbf$ is composed of independent Rademacher random variables, then:
    \[
    \Ebb \|\Xbf\|_{sp} \lesssim \sqrt[4]{\log n}\sqrt{m}.
    \]
    \item By \citet{tropp2015introduction}, if $\Xbf$ is composed of independnet zero-mean random variables, then:
    \[
    \Ebb \|\Xbf\|_{sp} \lesssim \max_i\sqrt{\sum_{j}\Ebb \Xbf_{i,j}^2} + \max_j\sqrt{\sum_{i}\Ebb \Xbf_{i,j}^2} + \sqrt[4]{\sum_{i,j}\Ebb\Xbf_{i,j}^4}
    \]
\end{itemize}
\end{lemma}

\begin{lemma}
\label{lemma:append-lemma-importance-weighting}
Let $P$ and $Q$ be the training and target distribution supported on $\Scal_P, \Scal_Q \subseteq \Dcal$, and $w(u,i) = \frac{Q(u,i)}{P(u,i)}$, for $(u,i)\in \Scal_P \cap \Scal_Q$. $\Dcal_n$ consists of training instances sampled i.i.d from $P$.
Given a single hypothesis $f\in\Fcal$, suppose $w_{ui}\in (0,1)$, for any $\delta>0$, it holds with probability at least $1-\delta$ that:
\begin{equation*}
\Ebb_{Q}R(f) \leq \Ebb_{P_{n,w}}R(f) + \frac{2\log\frac{1}{\delta}}{3n} + \sqrt{\frac{2d(P\|Q)\log\frac{1}{\delta}}{n}},
\end{equation*}
where $d(P\|Q) = \int_{\Scal_Q} (dP / dQ) dP$.
\end{lemma}

The above result for importance weighting of a single hypothesis is stated in the Theorem 1 of \cite{cortes2010learning}.

\section{Experiment details}
\label{sec:add-experiment}

Both MCF and NCF are implemented in Tensorflow 2.3, and the computation infrastructure involves a Nvidia Tesla V100 GPU with 32 Gb memory. We provide a kernel SVM implementation using the python Scikit-learn package, and a CVX implementation using the Python CVXOPT API. The code is also provided as a part of the supplementary material. 

As we mentioned in the main paper, we use the log loss $\ell(u) = \log(1+\exp(-u))$ for all our experiments. The metrics we consider, i.e. the ranking AUC, top-k hitting rate and NDCG are computed from a scan over all the possible candidates. Since there is only one relevant item (the last interacted item) for each user, then according to \citet{rendle2019evaluation}, the metric computations are simplified to:
\begin{itemize}
    \item Suppose the ranking of the relevant item, among the whole set of candidate items: $\tilde{\Ical}(u) = \Ical  - \{i\, |\, (u,i) \in \Dcal_{\text{train}}\}$, is given by $r$ for user $u$. Then the ranking AUC for user $u$ is given by: AUC(u)$= \big(|\tilde{\Ical}(u)| - r\big) / \big(|\tilde{\Ical}(u)| - 1\big)$;
    \item The top-k hitting rate for user $u$ is: HR@k(u)$= 1[r \leq k]$;
    \item the top-k NDCG for user $u$ is: NDCG@k(u)$= 1[r \leq k]\frac{1}{\log_2(r+1)}$.
\end{itemize}
Then the overall metric is computed by taking the population average.  

Since a large proportion of our discussion surrounds the gradient descent, we use the SGD optimizer unless otherwise specified. 

\textbf{Experiment for Figure 1.}

We point out that the data is relatively small after the subsampling so that the performance can vary significantly across different sampled datasets. Therefore, we do not repeat the experiments on different sampled datasets, but on the different splits (of generating negative samples). We point out that sampling the movies by popularity is necessary, because otherwise, the obtained data will be very sparse and thus not representative of the original dataset. After sampling 200 movies and 200 users, we end up with approximately 30,000 records for training when setting the number of negative samples to 4. Notice that this already requires a 30,000$\times$30,000 matrix for the kernel method.

The experiment setting follows that of the transductive CF, where the negative samples are constructed via sampling without replacement, and the random splits are conducted before training. As for the positive label, we adopt the standard setting where the last user-item interaction is used for testing, and the rest are used for training. 
Notice that we do not need validation data in this case, because there are no tuning parameters as we fixed the dimensions and learning rate, and do not use regularizations of any kind.   

\textbf{Experiment for Figure 2.}

We first use CVXOPT to obtain the exact solution of the convex nuclear-norm max-margin problem in Theorem 2. The optimality is reached with the duality gap $\leq 1e^{-19}$. We use the same dataset generated for the experiments in Figure 1. As we stated before, we consider the unscaled $N(0,0.1)$ initialization, set the moderate width of $d=32$ and the learning rate of $0.1$. Again, the repetitions are over the random splits of the negative samples. The normalized margin for the nuclear-norm max-margin problem is obtained via: $\gamma^{\text{SVM}}_{u,i} = y_{u,i}{\Xbf}_{u,i} / \|{\Xbf}\|_*$, and the normalized margin for MCF is obtained via: $\gamma^{\text{MCF}}_{u,i} = y_{u,i}\langle \zbf_u, \zbf_i \rangle \big/ \|\Zbf_U \Zbf_I^{\intercal}\|_*$.

\textbf{Experiment for Figure 3.}

We use all the data for the inductive CF task, where the last user-item interaction is used for testing, the second-to-last is used for validation, and the rest are used for training. The negative samples are also obtained via sampling without replacement, where we fix the number of negative samples to 4 for each positive interaction. Due to the sampling without replacement, setting the number of negative samples per positive to a high value may not increase the total number of negative samples proportionally (e.g. a user may have watched 300 out of the 1,000 movies). Therefore, we do not tune the number of negative samples per positive.  

We select $d$ from $\{16, 32, 48\}$ for MCF, and $d \in\{16, 32, 48\}, d1 \in \{8, 16, 24\}$ for NCF (since we study the two-layer setting). We experiment with a learning rate of $\{0.01, 0.05, 0.1, 0.2\}$, and do not find a significant difference since we study the converged behavior after several thousand epochs. For illustration purpose, we use $0.1$ as the learning rate. We make the hyper-parameter selection over one run and fix it during the rest repetitions. We find $d=32$ and $d=32, d_1=16$ gives the best performance for MCF and NCF, as we reported in Figure 3. The results reported in Figure 1 are the average over 10 random splits of the negative samples (and random initializations).

\textbf{Experiment for Figure 4 and 5.}

The learning of the relevance mechanism, exposure mechanism and the final data generating mechanism are stated in Section 6. When learning the relevance and exposure mechanism, we do not conduct the train/test split, since this step aims to construct the mechanisms according to the data, rather than examining how the models fit the data. When the $g_{\text{rel}}$ and $g_{\text{expo}}$ are given by the MCF, we use $d=32$; and when they are given by the NCF, we use $d=32, d_1=8$. We do not tune these hyperparameters due to the same reason stated above. We use the mean squared-root error (MSE) and the binary cross-entropy loss when training the relevance and exposure models. Unlike training for the CF tasks, we use the Adam optimizer with a learning rate of 0.001, which we find to work well with the MSE. 

After we settle down with the learnt relevance and exposure mechanism, we generate the observed data according to the click model. Before that, we tune the $\mu$ and $\rho$ in the relevance model to ensure the generated data has about the same sparsity as the original data. Since neither MCF nor NCF leverage the sequential information, the order by which we generate the interacted items for a specific user is not important. After we generate the click data for all the user-item pairs, we sample from the positive and negative parts with replacement to construct the training, validation and testing data, according to the empirical data distribution (which is a uniform distribution over the indices). 

All the results reported in Figure 4 and 5 about NCF are from the concatenation. We observe that NCF with addition has very similar patterns in the inductive CF experiments, so we do not report its results to avoid repetition. For MCF, we select $d\in\{16,32,48\}$, and for NCF we select $d\in\{16,32,48\}, d_1\in\{8,16,24\}$. We also do not experiment on using regularizations here. We repeat the generation, training, evaluation process for ten times. Each time, we tune the hyperparameters according to the validation performance. The evaluation metric we report is the biased and unbiased ranking AUC, where the biased AUC is computed in the regular fashion, and the unbiased AUC is computed via: unbiased-AUC(u,i)$= \big(|\tilde{\Ical}(u)| - r(i)\big) / \big(|\tilde{\Ical}(u)| - 1\big)\cdot \frac{1}{p(O_{u,i}=1)}$, where $r(i)$ is the ranking of item $i$ in $\tilde{\Ical}(u)$.

\end{document}